\documentclass[11pt]{article}
\usepackage[a4paper,margin=27mm,headheight=14pt]{geometry}
\usepackage[T1]{fontenc}
\usepackage{lmodern,amsmath,amssymb,amsthm,mathtools}
\usepackage{booktabs,array,tabularx,microtype,enumitem,needspace}
\usepackage{float,placeins}

\usepackage[hidelinks,pdfusetitle]{hyperref}
\usepackage{xurl}
\usepackage{fancyhdr}
\setlist{nosep,leftmargin=*}
\numberwithin{equation}{section}
\newtheorem{theorem}{Theorem}[section]
\newtheorem{proposition}[theorem]{Proposition}
\newtheorem{formalization}[theorem]{Formalization result}
\theoremstyle{definition}

\newcommand{\F}{\mathbb F}
\newcommand{\Z}{\mathbb Z}

\newcommand{\GL}{\mathrm{GL}}
\newcommand{\SL}{\mathrm{SL}}
\newcommand{\PGL}{\mathrm{PGL}}
\newcommand{\PSL}{\mathrm{PSL}}
\newcommand{\PSU}{\mathrm{PSU}}
\newcommand{\PSp}{\mathrm{PSp}}
\newcommand{\POm}{\mathrm{P}\Omega}
\newcommand{\Alt}{\operatorname{Alt}}
\newcommand{\Sym}{\operatorname{Sym}}
\newcommand{\Co}{\mathrm{Co}}
\newcommand{\Fi}{\mathrm{Fi}}
\newcommand{\Suz}{\mathrm{Suz}}
\newcommand{\Aut}{\operatorname{Aut}}
\newcommand{\atlas}{\textsc{Atlas}}
\newcommand{\code}[1]{\texttt{\small #1}}
\newcommand{\leanname}[1]{{\small\nolinkurl{#1}}}
\title{A constructive ATLAS\\of finite simple groups in Lean}
\author{Gerald H\"ohn\\[3pt]
  \small Department of Mathematics, Kansas State University\\
  \small\href{mailto:gerald@monstrous-moonshine.de}{\nolinkurl{gerald@monstrous-moonshine.de}}}
\hypersetup{pdfauthor={Gerald H\"ohn},
  pdfsubject={Constructions, orders, simplicity, and structural comparisons of finite simple groups},
  pdfkeywords={finite simple groups, sporadic groups, classical groups, octonions, Lean, formal verification}}
\date{25 September 2026}
\begin{document}
\maketitle

\begin{abstract}
We present a constructive atlas of finite simple groups with proofs of
their orders, simplicity, and structural properties. The eight completed
families are cyclic groups of prime order, alternating groups, the classical
series $A_r(q)$, $B_r(q)$, $C_r(q)$, $D_r(q)$, and the exceptional series $G_2(q)$
and the small Ree groups $ {}^2G_2(3^{2m+1})$, $m\geq1$.
The fifteen sporadic entries are $M_{11}$, $M_{12}$, $M_{22}$, $M_{23}$, $M_{24}$,
$\Co_1$, $\Co_2$, $\Co_3$, $\mathrm{McL}$, $\mathrm{HS}$, $\Suz$, $J_2$,
and $\Fi_{22}$, $\Fi_{23}$, $\Fi_{24}'$. The simple parameter ranges and
exceptional cases are stated explicitly. The models arise from codes,
lattices, forms, algebras, and finite geometries, including the split
octonions and the Conway--Parker algebra. They retain natural actions,
stabilizers, central quotients, and comparison maps for subsequent group
theory. Several classical comparison isomorphisms relate the models, while
involution-class counts distinguish the equal-order orthogonal and
symplectic families in odd characteristic and rank at least three.
Drawing on classical sources and companion mathematical work, the
project develops the construction side of finite simple group theory,
not the exhaustiveness proof of the classification. The completed models
and stated proofs are formalized in Lean for verification and reuse.
\end{abstract}

\smallskip
\noindent\textit{Source repository:}
\mbox{\url{https://github.com/Moonshine-in-Kansas/atlas}}

\section{Introduction}

The construction side of finite simple group theory asks for concrete
models together with proofs of their orders, simplicity, and
characteristic structure. Codes, lattices, algebras, forms, and finite
geometries provide such models and retain natural actions, stabilizers,
central quotients, and maps between related objects. The \atlas\ project
organizes these constructions into a structural library, with proofs
formalized in Lean~\cite{Lean4} using Mathlib~\cite{Mathlib}.

The catalogue reported here records the construction status of
18 September 2026: eight family packages and fifteen sporadic
constructions, summarized in Formalization
results~\ref{thm:families} and~\ref{thm:sporadics}. After the cyclic
and alternating groups, the classical series are presented in the order
$A,B,C,D$. Their explicit comparison isomorphisms relate models, while
involution-class counts separate equal-order orthogonal and symplectic
groups in odd characteristic and rank at least three. The full
automorphism group of the split octonions gives $G_2(q)$ uniformly in
all characteristics; singular-point geometry determines its order
and proves simplicity for $q>2$. Twisted seven-dimensional matrices
give the small Ree groups, whose point geometry determines their full
stabilizers and whose preserved tensors give embeddings into $G_2(q)$.
The sporadic constructions then follow the order of the abstract:
the Mathieu and Conway groups, $\mathrm{McL},\mathrm{HS},\Suz,J_2$,
and the three Fischer groups. Shared code, lattice, and algebra data
are introduced before the entries that use them.

Classification has two complementary components: constructing the
groups on the list and proving that every finite simple group occurs
there. \atlas\ develops the first component without assuming the
exhaustiveness theorem of the classification of finite simple groups
(CFSG). Construction, recognition, and exhaustiveness remain distinct
obligations. The actions, stabilizers, and comparison maps retained
here are intended as inputs to subsequent structure and recognition
theorems; Section~\ref{sec:contract} states the conditional recognition
questions precisely.

The mathematical sources combine established constructions with the
author's companion work. The Hall--Janko preprint~\cite{HallJankoNote}
and the Conway--Parker algebra manuscript~\cite{FischerAlgebra}
supply source mathematics for the corresponding entries. The present
article describes the models, their order and simplicity arguments,
and the structural relationships preserved in their formal
realization. The direct Monster-order work with
Seysen~\cite{MonsterOrder} supplies a starting point for future
formalization, not a completed entry in this release. Companion work
on integral Albert algebras and the rank-26 lattice and Cayley-plane
design develops the arithmetic and geometric direction
\cite{AlbertMassNote,CayleyPlaneNote}.

A broader objective is to explain why sporadic symmetries arise.
Vertex operator algebras (VOAs) provide a principal direction,
illustrated by the public Monster work in Section~\ref{sec:monster}
and by joint work with Lam on a VOA construction of the Fischer
groups~\cite{FischerVOA}, currently in preparation.
Retaining the mathematical objects is essential for this purpose:
the Fischer construction, for example, keeps its algebra, rays, and
incidence geometry together with the arguments determining the group
and its order. Formal verification makes these arguments and their
structural maps available for checked reuse.

The name recalls the \emph{Atlas of Finite Groups}~\cite{ATLAS85}.
Our initial scope is constructions and structural theorems; complete
character tables and maximal-subgroup lists are longer-term projects.
Sections~\ref{sec:elementary}--\ref{sec:fischer} present the
individual families and sporadic entries in catalogue order.
Section~\ref{sec:classical} also records the classical conventions
and comparison results. Sections~\ref{sec:monster}
and~\ref{sec:programme} distinguish further mathematical goals from
the completed catalogue. Related formalization projects are discussed
in Section~\ref{sec:related}. Technical implementation, declaration
interfaces, logical foundations, and verification records are collected
in Appendices~\ref{sec:libraries}--\ref{app:provenance}, including the
complete per-group source-line dependency table.

\section{The mathematical content of an entry}\label{sec:contract}

For a family label $X$, let $\mathcal A_X$ be an explicitly specified
parameter domain. An entry consists of a defined group $G_X(\lambda)$,
for $\lambda\in\mathcal A_X$, and proofs of
\begin{equation}\label{eq:entry}
 \lvert G_X(\lambda)\rvert=N_X(\lambda),\qquad
 G_X(\lambda)\text{ is simple},\qquad
 P_X\bigl(G_X(\lambda)\bigr),
\end{equation}
together with finiteness and, for the nonabelian entries, explicit
noncommutativity. Here $P_X$ describes independently defined mathematical
structure: an action, a code, a lattice, an algebra, or specified maps
between such objects. Neither $\mathcal A_X$ nor $P_X$ may incorporate the
desired order or simplicity assertion as an unproved premise. For cyclic
groups of prime order, commutativity is part of the intended conclusion.

For some entries, we also compute $k_2(G)$, the number of conjugacy
classes of nonidentity involutions, to distinguish nonisomorphic groups
of the same order; see Section~\ref{sec:equalorder}.

Here \emph{construction} refers to a specified mathematical model and
proofs of its properties. The development uses classical mathematics in
Lean, including classical choice; it need not provide an executable
enumeration of every group element.

\subsection{The chosen object and its full symmetry group}

For a symmetry-group model, equality between a generated subgroup and
the full automorphism group is a mathematical theorem. In a code or
lattice construction, generation or a stabilizer argument establishes
this equality before the order and simplicity assertions are attached
to the full group.

Central quotients come with their maps and kernels. Thus $\Co_0$ is
distinguished from $\Co_1$, and a Hermitian isometry group from its
scalar quotient. A different distinction is that between $\Fi_{24}$
and its derived subgroup $\Fi_{24}'$. Abstract group isomorphisms,
subgroup equalities, and equivariant identifications of actions are
recorded with the corresponding data.

\subsection{Conditional recognition and comparison}

Conditional on CFSG, order determines a finite simple group up to
isomorphism except for the pairs
$\{\Alt_8\cong\PSL_4(2),\PSL_3(4)\}$ and
$\{B_r(q),C_r(q)\}$ with $r\geq3$ and $q$ odd~\cite{OrderRecognition}.
These pairs are distinguished by the number $k_2(G)$ of conjugacy
classes of nonidentity involutions: the first has values $2$ and $1$,
and the second $r$ and $\lfloor r/2\rfloor+1$~\cite{Taylor,BorovikInvolutions}.
Consequently, assuming CFSG, the pair $(|G|,k_2(G))$ determines the
isomorphism type of a finite simple group. The formal
$B/C$ counts and both nonisomorphism results are described in
Section~\ref{sec:equalorder}; the general recognition theorem and the
$\Alt_8/\PSL_3(4)$ involution counts are proposed work.

For construction proofs, the relevant independence question concerns
the hypotheses and dependency graph of each theorem. Background results
and earlier constructions can be reused. A classification argument can
also be non-circular when the property being deduced was not used in
that argument, for example when it concerns smaller groups.
The Monster order illustrates why precise attribution matters:
Griess--Meierfrankenfeld--Segev, Corollary~3.7.3, uses substantial
information about other sporadic groups~\cite{GMS}, whereas
H\"ohn--Seysen~\cite{MonsterOrder} gives a direct computational
determination of the Monster and Baby Monster orders in the
Leech-lattice setting. The former dependency is more specific than an
invocation of the entire CFSG.

Recognition through an action or local geometry offers a natural
interface: a classification argument produces the structure, and a
comparison theorem identifies it with a constructed model. The actions,
stabilizers, and maps retained here are intended to support such theorems.

\section{Current constructions}\label{sec:inventory}

The catalogue records eight completed family packages and fifteen
sporadic constructions~\cite{AtlasSoftware}. In this article,
\emph{completed} means that the stated order, simplicity, and structural
results have the recorded proofs for the specified parameter
domains; it does not mean that every property of the groups has been
formalized. The verification scope is described in
Appendix~\ref{sec:workflow}. Family labels count construction packages
and can overlap up to isomorphism. Section~\ref{sec:classical} records
the status of the comparison maps.

\Needspace{6\baselineskip}
\begin{formalization}[Family packages]\label{thm:families}
The following family packages are completed. Here $n$ is the degree of
an alternating group or the matrix dimension of $\PSL_n$; $r$ denotes
Dynkin rank in the $B,C,D$ families.
For each finite field $F$ below, put $q=|F|$.
\begin{enumerate}[label=\textup{(\roman*)}]
\item For every prime $p$, the cyclic group $C_p$ has order $p$ and is
simple and commutative.
\item For every $n\geq5$, the alternating group $\Alt_n$ is nonabelian
simple and has order $n!/2$.
\item Let $F$ be any finite field, put $q=|F|$, and let $n\geq2$.
The central quotient
$\PSL_n(F)=\SL_n(F)/Z(\SL_n(F))$ has order
\begin{equation}\label{eq:pslorder}
 |\PSL_n(F)|=
 \frac{q^{n(n-1)/2}\prod_{i=2}^{n}(q^i-1)}{\gcd(n,q-1)}.
\end{equation}
It is simple if and only if
$(n,q)\notin\{(2,2),(2,3)\}$. For those admissible parameters it is
also nonabelian.
\item For every finite field $F$ and rank $r\geq1$, the projective
quadratic model $B_r(F)$ has the order in~\eqref{eq:bcorder}.
It is simple exactly outside $(r,q)=(1,2),(1,3),(2,2)$, and
nonabelian at the simple parameters. The full nonsimple model at
$(r,q)=(2,2)$ and its derived subgroup are distinguished.
\item For every finite field $F$ and rank $r\geq1$, the symplectic
quotient $C_r(F)=\PSp_{2r}(F)$ has the order in~\eqref{eq:bcorder}.
It is simple exactly outside $(r,q)=(1,2),(1,3),(2,2)$, and
nonabelian at the simple parameters. The full nonsimple model at
$(r,q)=(2,2)$ and its derived subgroup are distinguished. The
symplectic rank-zero model is trivial.
\item For every finite field $F$ and $r\geq4$, the projective
split-orthogonal model $D_r(F)$ has the order in~\eqref{eq:dorder}
with $\varepsilon=1$ and is nonabelian simple in every characteristic,
including $D_4(2)$.

\item For every finite field $F$, with $q=|F|$, let
$G_2(F)=\Aut_F(\mathbb O_{\mathrm{split}}(F))$ be the full group of
multiplication-preserving linear automorphisms of the constructed
split octonion algebra. Then
\begin{equation}\label{eq:g2order}
 |G_2(F)|=q^6(q^6-1)(q^2-1).
\end{equation}
It is simple if and only if $q>2$, and noncommutativity is proved
for every $q$. At $q=2$, an explicitly defined normal subgroup has
order $6048$ and index two; its own simplicity and identification
with $\PSU_3(3)$ are not included in this result.
\item Let $F$ be a finite field of characteristic three with
$|F|=q=3^{2m+1}$, where $m\geq1$. The group $ {}^2G_2(F,m)$ generated
by the Tits-twisted matrices in $\GL_7(F)$ is nonabelian
simple and has order
\begin{equation}\label{eq:reeorder}
 |{}^2G_2(F,m)|=q^3(q^3+1)(q-1).
\end{equation}
It acts faithfully and doubly transitively on a constructed set of
$q^3+1$ points. The full point stabilizer is the soluble Borel subgroup
of order $q^3(q-1)$. Field transport and an injective homomorphism into
the split-octonion group $G_2(F)$ are proved. The parameter $q=3$
and a graph-automorphism fixed-point characterization are not part of
this assertion.
\end{enumerate}
\end{formalization}

The orthogonal packages retain intrinsic kernels, full centers,
root generation, singular-line geometry, and field transport. Their
domains include comparison models outside the primary label ranges in
Table~\ref{tab:classical}. For type $A$ and $G_2$, finite-field witnesses
give existence at every admissible prime power; for the small Ree
groups, a canonical field witness is supplied for every $m\geq1$.
The notation $ {}^2G_2(F,m)$ follows the Lean model's explicit twist
parameter; $|F|=3^{2m+1}$ determines $m$ uniquely.

Cyclic and alternating entries reuse Mathlib's principal order and
simplicity results. The remaining packages combine general library
theorems with the constructions described below. Mathieu parameters
such as a dodecad or marked coordinates select objects in constructed
finite structures; standard choices are supplied.

\begin{formalization}[Sporadic packages]\label{thm:sporadics}
The development contains finite nonabelian simple groups with the
constructions and orders in Table~\ref{tab:sporadic}. These entries include
the specified code, lattice, stabilizer, scalar-quotient, or algebraic interfaces.
\end{formalization}

\begin{table}[!htbp]
\centering\small
\setlength{\tabcolsep}{5pt}
\begin{tabularx}{\textwidth}{@{}l r X@{}}
\toprule
Group & Order & Chosen construction\\
\midrule
$M_{11}$ & $7\,920$ & Point stabilizer in the dodecad model of $M_{12}$\\
$M_{12}$ & $95\,040$ & Golay dodecad stabilizer\\
$M_{22}$ & $443\,520$ & Pointwise stabilizer of two Golay coordinates\\
$M_{23}$ & $10\,200\,960$ & Stabilizer of one Golay coordinate\\
$M_{24}$ & $244\,823\,040$ & Full binary Golay code automorphism group\\
\addlinespace
$\Co_1$ & $4\,157\,776\,806\,543\,360\,000$ & Leech isometry group modulo $\{\pm I\}$\\
$\Co_2$ & $42\,305\,421\,312\,000$ & Stabilizer of a Leech vector of squared norm $4$\\
$\Co_3$ & $495\,766\,656\,000$ & Stabilizer of a Leech vector of squared norm $6$\\
$\mathrm{McL}$ & $898\,128\,000$ & Pointwise stabilizer of a Leech $2$--$2$--$3$ triangle\\
$\mathrm{HS}$ & $44\,352\,000$ & Pointwise stabilizer of a Leech $2$--$3$--$3$ triangle\\
\addlinespace
$\Suz$ & $448\,345\,497\,600$ & Full Eisenstein Hermitian group modulo six scalars\\
$J_2$ & $604\,800$ & Full icosian Hermitian group modulo two signs\\
\addlinespace
$\Fi_{22}$ & $64\,561\,751\,654\,400$ & Common centralizer of two marked commuting involutions modulo their subgroup\\
$\Fi_{23}$ & $4\,089\,470\,473\,293\,004\,800$ & Centralizer of a marked involution modulo that involution\\
$\Fi_{24}'$ & $1\,255\,205\,709\,190\,661\,721\,292\,800$ & Parity kernel in the constructed Fischer ray group\\
\bottomrule
\end{tabularx}
\caption{The fifteen completed sporadic constructions. A triangle label records half the three squared edge lengths.
The groups are stabilizers of vectors or configurations, with the
pointwise and setwise distinctions retained.}
\label{tab:sporadic}
\end{table}

The auxiliary groups, notably $\Co_0$ and the retained scalar covers,
are not counted as additional simple groups. Nor are small isomorphisms
between family members counted as new entries. The eleven sporadics outside that completed catalogue are
$J_1,J_3,J_4,\mathrm{He},\mathrm{Ru},\mathrm{O}'\mathrm{N},
\mathrm{HN},\mathrm{Ly},\mathrm{Th},B$, and the Monster.
The classical extension is discussed in Section~\ref{sec:classical},
and the remaining uniform Lie-type programme in Section~\ref{sec:programme}. The Tits group is a separate exceptional Lie-type entry,
not an additional sporadic group.

\section{Cyclic and alternating groups}\label{sec:elementary}

\subsection{Cyclic groups of prime order}

For each prime $p$, the cyclic entry is the standard group $C_p$ of
order $p$. Its order, simplicity, and commutativity are the first
assertions of Formalization result~\ref{thm:families}. The construction
reuses Mathlib's cyclic-group results; the declaration interfaces are
recorded in Appendix~\ref{sec:libraries}. These are the commutative
entries of the catalogue, kept separate from the nonabelian families.

\subsection{Alternating groups}

For every $n\geq5$, the alternating entry is $\Alt_n$, with its natural
permutation action, order $n!/2$, and nonabelian simplicity. Here again
the principal order and simplicity results are supplied by Mathlib.
The degree convention and range are those of Formalization
result~\ref{thm:families}; low-degree groups outside that range are
not additional entries. The comparison maps with classical groups,
including $\Alt_5$, $\Alt_6$, and $\Alt_8$, are treated in
Section~\ref{sec:classical} after the classical models have been
introduced.

\section{The classical series \texorpdfstring{$A_r(q)$}{Ar(q)}: projective special linear groups}
\label{sec:psl}

For $A_r(q)=\PSL_{r+1}(q)$, write $n=r+1$ for the matrix dimension.
The projective linear construction starts from matrix groups and
central quotients, followed by parameterized generation and
simplicity arguments. Its Mathlib inputs are recorded in
Appendix~\ref{sec:libraries}.

\subsection{Order and the determinant quotient}

The starting point is the standard count of ordered bases,
\[
 |\GL_n(q)|=\prod_{i=0}^{n-1}(q^n-q^i).
\]
The scalar center of $\GL_n(q)$ has order $q-1$. The center of
$\SL_n(q)$ consists of the scalar matrices $\lambda I$ with
$\lambda^n=1$ and therefore has order $d=\gcd(n,q-1)$, by cyclicity
of $\F_q^{\times}$. These counts give~\eqref{eq:pslorder}.
The development also proves a divisibility-safe identity
\[
 |\PSL_n(q)|(q-1)d=\prod_{i=0}^{n-1}(q^n-q^i),
\]
so that the natural-number division in the displayed order formula is
known to be exact.

The determinant gives a stronger structural result. Over any field $F$
and in positive dimension there is an isomorphism
\begin{equation}\label{eq:detquot}
 \PGL_n(F)/\PSL_n(F)\ \cong\ F^{\times}/(F^{\times})^n,
\end{equation}
where $\PSL_n(F)$ on the left denotes its canonical image in $\PGL_n(F)$.
Scaling a matrix changes its determinant by an $n$th power, so the map
is well defined. If $\det(g)=u^n$, then $u^{-1}g$ has determinant one
and represents the same projective transformation; this identifies the
kernel. Surjectivity follows from that of the determinant.
The formalized equivalence includes compatibility with matrix
representatives. It is consequently usable independently of the finite
order calculation.

\subsection{Generation and simplicity}

Elementary transvections are written $x_{ij}(a)=I+aE_{ij}$.
The diagonal/transvection reduction available in Mathlib is completed
by explicit two-coordinate identities, giving generation of $\SL_n(F)$
by transvections. For distinct $i,j,k$, the commutator identity
\begin{equation}\label{eq:comm}
 [x_{ik}(a),x_{kj}(b)]=x_{ij}(ab),
 \qquad [g,h]=ghg^{-1}h^{-1},
\end{equation}
proves perfectness when $n\geq3$. Generation and perfectness pass to
the projective quotient.

For a line $L=Fv$, consider the subgroup of transformations
\[
 u_{\varphi}(x)=x+\varphi(x)v,\qquad \varphi(v)=0.
\]
It is abelian, is normal in the stabilizer of $L$, and transforms
covariantly under conjugation. Its projective images, as $L$ varies,
generate $\PSL_n(F)$. Together with the faithful primitive projective
action, these facts fit the following criterion.

\begin{proposition}[Iwasawa criterion]\label{prop:iwasawa}
Let a nontrivial perfect group $G$ act faithfully and primitively on a
set $\Omega$ with at least two points. Suppose that for some
$\alpha\in\Omega$, the stabilizer $G_\alpha$ has an abelian normal
subgroup $A$ whose conjugates generate $G$. Then $G$ is nonabelian simple.
\end{proposition}

\begin{proof}
Let $1\ne N\lhd G$. The $N$-orbits form a $G$-invariant partition.
Faithfulness excludes the partition into singleton orbits for such $N$,
and primitivity therefore makes $N$ transitive. Hence $G=NG_\alpha$.
In $G/N$, the image of every conjugate of $A$ is consequently the image
of $A$ itself, because $A\lhd G_\alpha$. The conjugates generate $G$,
so $G/N$ is abelian. Perfectness implies $G/N=1$, and $N=G$.
A nontrivial perfect group cannot be abelian.
\end{proof}

The existing formal criterion and projective-action infrastructure are
reused. The added general-rank generation and perfectness arguments
discharge its hypotheses. This proves simplicity for every field in
dimension at least three. For finite fields in dimension two, the
existing Mathlib theorem supplies simplicity when $q\geq4$.
The two exceptional cases are proved nonsimple using their projective
actions and orders. The final family statement gives the exact exception set. Numerical
examples follow by specialization.

\section{The classical series \texorpdfstring{$B,C,D$}{B, C, D} and comparison isomorphisms}\label{sec:classical}

The untwisted families $A,B,C,D$ have completed packages. Following
the type-$A$ construction in Section~\ref{sec:psl}, we describe the
$B,C,D$ models in that order before giving their comparisons.
Throughout, $q=|F|=p^f$ for a finite field $F$, and $r$ denotes
\emph{Dynkin rank}, not matrix dimension. We write $\Alt_m$ for an
alternating group to distinguish it from the Lie-type symbol $A_r(q)$.

\subsection{The odd-dimensional quadratic series \texorpdfstring{$B_r(q)$}{Br(q)}}\label{sec:type-b}

The quadratic construction starts from
\[
 Q_B(x,y,z)=\sum_{i=1}^{r}x_i y_i+z^2.
\]
It constructs the subgroup generated by Siegel transformations and
divides by its scalar subgroup, proved to be its full center. In the
completed parameter ranges, the elementary subgroup is identified in
odd characteristic with the determinant/spinor kernel. In
characteristic two it is the full quadratic isometry group: the
odd-dimensional space has one-dimensional polar radical, and its
symplectic comparison is proved.

The geometric package has a faithful singular-line action in every
positive rank, transitivity over finite fields, and primitivity from
rank two. Its order is the $B_r(q)$ formula in~\eqref{eq:bcorder};
the exact simplicity exceptions are those in Formalization
result~\ref{thm:families}. The odd-characteristic $B_2$ construction
uses its symplectic $C_2$ image, and the geometric proof uses transport
through a $B_1$ complement. Thus the maps $B_1\cong A_1$ and
$B_2\cong C_2$ are part of the construction, even though the
symplectic series is presented next.

\subsection{The symplectic series \texorpdfstring{$C_r(q)$}{Cr(q)}}\label{sec:type-c}

The model is the symplectic central quotient $C_r(F)=\PSp_{2r}(F)$.
Its order agrees with that of $B_r(F)$, as in~\eqref{eq:bcorder}.
For positive rank it is simple exactly outside
$(r,q)=(1,2),(1,3),(2,2)$ and is nonabelian at the simple parameters;
the rank-zero model is trivial. The full nonsimple rank-two binary
model is kept distinct from its derived subgroup.

The retained comparisons identify the rank-one model with
$\PSL_2(F)$, the rank-two model with $B_2(F)$ over every finite
field, and the higher-rank models with $B_r(F)$ in characteristic
two. In odd characteristic and rank at least three, the same order
does not give an isomorphism: Section~\ref{sec:equalorder} separates
the two series by their involution-class counts.

\subsection{The split-orthogonal series \texorpdfstring{$D_r(q)$}{Dr(q)}}\label{sec:type-d}

The split quadratic construction starts from
\[
 Q_D(x,y)=\sum_{i=1}^{r}x_i y_i.
\]
As in type $B$, the subgroup generated by Siegel transformations is
divided by its scalar subgroup, proved to be its full center. In the
completed parameter ranges it is the determinant/spinor kernel in
odd characteristic and the intrinsic Dickson kernel in characteristic
two. The model retains a faithful primitive singular-line action in
rank at least four, together with root generation and field transport.

For every finite field and $r\geq4$, the projective group has the
order in~\eqref{eq:dorder} with $\varepsilon=1$ and is nonabelian
simple, including $D_4(2)$. The auxiliary models $D_2,D_3$ are
retained for the product and exterior-square comparisons below;
they are not additional primary family labels.

\subsection{Rank conventions and order formulae}

The unitary $ {}^2A$ and minus-orthogonal $ {}^2D$ families and a
common Chevalley framework are further construction tasks.
The superscript in $ {}^2A_r$ denotes twisting; it is not a power of
$A_r$. In twisted notation, $q$ is the size of the fixed field:
$ {}^2A_r(q)=\PSU_{r+1}(q)$ is constructed using a Hermitian space over
$\F_{q^2}$. References writing $ {}^2A_r(q^2)$ denote the same family
with a different field convention. We use the fixed-field convention
also for $ {}^2D_r(q)$, $ {}^2E_6(q)$, and $ {}^3D_4(q)$.

\begin{table}[!htbp]
\centering\small
\renewcommand{\arraystretch}{1.12}
\begin{tabularx}{\textwidth}{@{}l l l X@{}}
\toprule
Type & Classical model & Rank range & Exclusions for simplicity\\
\midrule
$A_r(q)$ & $\PSL_{r+1}(q)$ & $r\geq1$ & $(r,q)=(1,2),(1,3)$\\
$B_r(q)$ & $\POm_{2r+1}(q)$ & $r\geq2$ & $(r,q)=(2,2)$\\
$C_r(q)$ & $\PSp_{2r}(q)$ & $r\geq3$ & None in this rank range\\
$D_r(q)$ & $\POm^+_{2r}(q)$ & $r\geq4$ & None in this rank range\\
$ {}^2A_r(q)$ & $\PSU_{r+1}(q)$ & $r\geq2$ & $(r,q)=(2,2)$\\
$ {}^2D_r(q)$ & $\POm^-_{2r}(q)$ & $r\geq4$ & None in this rank range\\
\bottomrule
\end{tabularx}
\caption{Classical families and proposed primary rank conventions,
not the full implemented parameter domains. Types $A,B,C,D$ have completed packages. The unitary and
minus-orthogonal families remain to be constructed.
The even-characteristic $B/C$ comparison is proved for the retained models.}
\label{tab:classical}
\end{table}

The symbols denote central quotients. Exceptional small parameters
require conventions for $\Omega$ as in
Taylor~\cite[Chapters~8, 10--12]{Taylor}.

The proposed primary ranges $r\geq1,2,3,4$ for $A,B,C,D$ remove
low-rank repetitions while keeping the comparison models in the
library. In particular, $B,C$ include all positive ranks, and $D_2,D_3$
are available. Even-characteristic $B/C$ isomorphisms and isolated
coincidences still give overlaps between primary labels.

For reference, the classical order formulae are
\begin{align}
 |A_r(q)|&=\frac{q^{r(r+1)/2}}{(r+1,q-1)}
             \prod_{i=2}^{r+1}(q^i-1),\notag\\
 |B_r(q)|=|C_r(q)|&=\frac{q^{r^2}}{(2,q-1)}
             \prod_{i=1}^{r}(q^{2i}-1),\label{eq:bcorder}\\
 |D_r^{\varepsilon}(q)|&=
       \frac{q^{r(r-1)}(q^r-\varepsilon)}{(4,q^r-\varepsilon)}
             \prod_{i=1}^{r-1}(q^{2i}-1),\label{eq:dorder}\\
 |{}^2A_r(q)|&=\frac{q^{r(r+1)/2}}{(r+1,q+1)}
             \prod_{i=2}^{r+1}(q^i-(-1)^i).\notag
\end{align}
Here $(a,b)=\gcd(a,b)$, $D_r^+=D_r$, and
$D_r^-={}^2D_r$~\cite{Taylor,Wilson}. The $A,B,C$ and split-$D$
formulae are proved for the models in the ranges stated above;
the unitary and minus-$D$ formulae are targets for subsequent development.

\subsection{Uniform coincidences and low ranks}

Table~\ref{tab:uniform-isos} records the uniform identifications of
central quotients, including auxiliary models below the primary rank
ranges. Here $ {}^2A_1$ abbreviates $\PSU_2$; type $A_1$ has no
nontrivial Dynkin-diagram twist.

\begin{table}[!htbp]
\centering\small
\renewcommand{\arraystretch}{1.15}
\begin{tabularx}{\textwidth}{@{}p{0.52\textwidth} X@{}}
\toprule
Isomorphism & Parameters and formalization status\\
\midrule
$A_1(q)\cong B_1(q)\cong C_1(q)\cong\PSU_2(q)$
 & $B_1\cong A_1$ and $C_1\cong A_1$ proved over every field;
   the $\PSU_2$ link remains planned. Simple for $q\geq4$.\\
$B_2(q)\cong C_2(q)$
 & Proved over every finite field, including the nonsimple $q=2$ case.\\
$D_2(q)\cong A_1(q)\times A_1(q)$
 & Proved over every field; the split rank-two model is not a simple entry.\\
$D_3(q)\cong A_3(q)$
 & Proved over every finite field for the split six-dimensional model.\\
$ {}^2D_2(q)\cong A_1(q^2)$
 & All $q$ mathematically; four-dimensional minus-type comparison still planned.\\
$ {}^2D_3(q)\cong{}^2A_3(q)$
 & All $q$ mathematically; six-dimensional minus-type comparison still planned.\\
$B_r(2^f)\cong C_r(2^f)$
 & Proved for $r\geq1$, $f\geq1$; primary-label overlaps persist for $r\geq3$.\\
\bottomrule
\end{tabularx}
\caption{Uniform identifications and their current proof status.
For $D_2(2)$ we retain the central quotient of the two $\SL_2(2)$
factors, not a possibly different derived-full-orthogonal convention.
The table includes comparison models outside the primary rank ranges.}
\label{tab:uniform-isos}
\end{table}

The exterior-square representation produces the six-dimensional
orthogonal model from a four-dimensional linear space. Related
Klein-correspondence arguments give low-rank comparisons. In
characteristic two, quotienting the polar radical of an odd-dimensional
quadratic space gives a symplectic space. Over a finite field this has
a converse lift~\cite[Chapters~11--12]{Taylor}. The completed untwisted
comparisons implement these constructions.

\subsection{Isolated coincidences across the classification list}

Table~\ref{tab:exceptional-isos} lists the isolated alternating and
Lie-type coincidences, including simple derived groups at small
parameters~\cite{ATLAS85,Wilson,ATLASOnline}. Its last column distinguishes
proved maps from known isomorphisms awaiting formal comparison.

\begin{table}[!htbp]
\centering\small
\setlength{\tabcolsep}{4pt}
\begin{tabularx}{\textwidth}{@{}r >{\raggedright\arraybackslash}p{0.43\textwidth} >{\raggedright\arraybackslash}X@{}}
\toprule
Order & Isomorphic simple groups & Status of comparison maps\\
\midrule
$60$ & $\Alt_5\cong A_1(4)\cong A_1(5)$ &
Both links with $\Alt_5$ proved.\\
\addlinespace
$168$ & $A_1(7)\cong A_2(2)$ & Proved through a common eight-point action.\\
\addlinespace
$360$ & $\Alt_6\cong A_1(9)$\newline
$\phantom{\Alt_6}\cong B_2(2)'\cong C_2(2)'$ &
All displayed links proved; the derived-group maps are explicitly exported.\\
\addlinespace
$504$ & $A_1(8)\cong{}^2G_2(3)'$ &
The $q=3$ derived-subgroup package and this comparison remain to be formalized.\\
\addlinespace
$6\,048$ & $ {}^2A_2(3)\cong G_2(2)'$ &
An index-two normal subgroup of $G_2(2)$ of order $6048$ is constructed;
its simplicity and the unitary comparison remain open in Lean.\\
\addlinespace
$20\,160$ & $\Alt_8\cong A_3(2)$ & Proved through the binary deleted permutation quadratic module.\\
\addlinespace
$25\,920$ & $ {}^2A_3(2)\cong B_2(3)\cong C_2(3)$ &
$B_2\cong C_2$ proved; the unitary link remains planned.\\
\bottomrule
\end{tabularx}
\caption{Isolated cross-series coincidences and their formal status,
up to the uniform identifications in Table~\ref{tab:uniform-isos}.
A prime denotes the derived subgroup. The status column records the
available comparison maps.}
\label{tab:exceptional-isos}
\end{table}

The binary quadratic construction gives $B_2(2)\cong\Sym_6$ and
$B_2(2)'\cong\Alt_6$.
The full group has order $720$ and its derived group order $360$.
The proved $B_2/C_2$ map identifies the full binary
symplectic model, and its restriction to the derived subgroups
is exported. The separate $\PSL_2(9)\cong\Alt_6$ comparison uses
the projective line and the ten partitions of six letters into two
triples. Composing these proved maps completes the order-$360$
row.

For the other small parameters, $G_2(2)'$ has index two in $G_2(2)$,
and $ {}^2G_2(3)'$ has index three in $ {}^2G_2(3)$.
The Tits group $ {}^2F_4(2)'$ has index two in $ {}^2F_4(2)$ but gives
no further coincidence in the table. The 26 sporadic groups supply no
additional coincidences with each other or with the alternating and
Lie-type simple groups~\cite{ATLAS85,Wilson,ATLASOnline}.
Prime cyclic groups are separate; $\Alt_3\cong C_3$ lies outside the
chosen alternating range. These tables do not classify nonsimple
groups or central covers, and their exhaustiveness is not a Lean
result here.

\subsection{Equal orders and proved nonisomorphisms}\label{sec:equalorder}

The order collisions in Section~\ref{sec:contract} require
\emph{nonisomorphism} certificates, not additional comparison
isomorphisms. Write $k_2(G)$ for the number of conjugacy classes of
nonidentity involutions, not the number of involutions themselves.
For every finite field of odd order and every positive rank, the
formalized involution-class theorems give
\begin{equation}\label{eq:k2bc}
 k_2(B_r(q))=r,\qquad
 k_2(C_r(q))=\lfloor r/2\rfloor+1.
\end{equation}
Consequently, for $r\geq3$ the constructed models satisfy
\[
 |B_r(q)|=|C_r(q)|,
 \qquad B_r(q)\not\cong C_r(q).
\]
The nonisomorphism follows from the isomorphism invariance of $k_2$.

\begin{table}[!htbp]
\centering\small
\begin{tabularx}{\textwidth}{@{}l l X@{}}
\toprule
Equal-order pair & $k_2$ values & Formal status\\
\midrule
$B_r(q),\ C_r(q)$ & $r,\ \lfloor r/2\rfloor+1$ &
Order equality and nonisomorphism proved for odd $q$, $r\geq3$.\\
\addlinespace
$\Alt_8,\ \PSL_3(4)$ & $2,\ 1$ &
Nonisomorphism proved via Sylow $2$-subgroup centers;
these $k_2$ counts are not certified.\\
\bottomrule
\end{tabularx}
\caption{The exceptional equal-order pairs. In the second row, the
$k_2$ values are known mathematically, while the formal nonisomorphism
uses a different invariant. Conditional
on CFSG these two kinds exhaust the nonisomorphic simple-group
order collisions~\cite{OrderRecognition,Taylor,BorovikInvolutions}.}
\label{tab:noniso}
\end{table}

The proved order-$20\,160$ obstruction constructs the upper
unitriangular subgroups in $\SL_4(2)$ and $\SL_3(4)$ and embeds them
in the corresponding projective groups. Both images have order $64$
and are Sylow $2$-subgroups; symbolic commutation formulas show that
their centers have orders $2$ and $4$, respectively. An isomorphism
of the ambient groups would preserve this invariant. Transport through
the independently proved $\PSL_4(2)\cong\Alt_8$ map gives the
alternating-group obstruction.
The optional involution-class counts remain separate: $\Alt_8$ has
the two even cycle types $2^2 1^4$ and $2^4$, whereas $\PSL_3(4)$
has one involution class~\cite[Chapter~4]{Taylor}.
Both nonisomorphism routes are unconditional, but neither a complete
catalogue-wide collision theorem nor the conditional recognition
theorem for $(|G|,k_2(G))$ is claimed as a Lean result.

\subsection{Comparison maps and remaining targets}

The comparisons marked as proved in Tables~\ref{tab:uniform-isos}
and~\ref{tab:exceptional-isos} are explicit multiplication-preserving
maps with proved inverses, rather than identifications by order.
Field isomorphisms also induce isomorphisms of the corresponding
orthogonal models.

The five isolated comparisons between the implemented families use
natural actions or the deleted permutation quadratic module. Their
order-$60$ argument also gives an elementary recognition theorem for
nonabelian simple groups of that order. Unitary, twisted-$D$, and the
remaining exceptional links are further targets. Each comparison of
actions or central covers needs its own maps, kernels, and equivariance
statement; the group equivalences alone concern the simple quotients.

\section{The exceptional family \texorpdfstring{$G_2(q)$}{G2(q)} from split octonions}\label{sec:g2}

The first completed exceptional Lie-type family is constructed directly
from the split octonion algebra. The formal model retains the algebra,
its full linear automorphism group, and its singular-point geometry;
see~\cite[Chapter~4]{Wilson} for the classical background.
The uniform proof includes characteristics two and three.

\subsection{The algebra and its full automorphism group}

Write $\mathbb O_{\mathrm{split}}(F)=\bigoplus_{i=1}^{8}F x_i$.
The multiplication is the bilinear extension of the fixed integral
split-octonion table used in the construction. In these coordinates,
for $a=\sum a_i x_i$,
\[
 1=x_4+x_5,\qquad \operatorname{tr}(a)=a_4+a_5,\qquad
 N(a)=a_1a_8+a_2a_7+a_3a_6+a_4a_5.
\]
Conjugation is $\bar a=\operatorname{tr}(a)1-a$. The formal development
proves both identity laws, bilinearity, the quadratic identity
$a^2-\operatorname{tr}(a)a+N(a)1=0$, the conjugate-product identities,
norm multiplicativity, and reversal of products under conjugation.
These are polynomial identities over arbitrary commutative rings.
The same algebra therefore applies in characteristic two.

The group $G_2(F)$ is the full subgroup of $\GL(\mathbb O_{\mathrm{split}}(F))$
preserving multiplication. Preservation of the identity, trace, norm,
and conjugation is proved from this definition. For example,
surjectivity forces the identity to be fixed; for a nonscalar $a$,
comparison of its quadratic identity before and after applying an
automorphism forces the trace and norm coefficients to agree.
Thus the invariants follow from preservation of multiplication. The natural eight-dimensional representation
is faithful, and coefficientwise field isomorphisms transport the
group and its action.

\subsection{Explicit subgroups, full stabilizers, and order}

Six root transformations are defined and verified directly from the
multiplication table. Their additive laws and parameter injectivity
are proved. Ordered products give a bijection $F^6\longrightarrow U$
with explicit multiplication and inverse formulas. The diagonal torus
$T$ is identified with $F^{\times}\times F^{\times}$, normalizes $U$,
and intersects it trivially. Consequently,
\[
 |U|=q^6,\qquad |T|=(q-1)^2,\qquad |UT|=q^6(q-1)^2.
\]
Together with explicit Weyl transformations, these maps are proved
to generate the full octonion automorphism group.

The singular vectors are the nonzero $a$ with
$\operatorname{tr}(a)=N(a)=0$. Their number is $q^6-1$.
Indeed, after eliminating one coordinate using the trace equation,
the count reduces to pairs of three-vectors and a scalar satisfying
$u\cdot v=t^2$: for each $u\ne0$ and each $t$, the $v$-fibre has
$q^2$ elements; for $u=0$, one has $t=0$ and $q^3$ choices of $v$.
Including zero gives $q^6$ vectors. The projective singular-point set
$\mathcal P$ therefore has size
\[
 |\mathcal P|=\frac{q^6-1}{q-1}.
\]

Coordinate reductions prove transitivity and determine the full
stabilizers. The pointwise stabilizer of the opposite vectors $x_1,x_8$
is isomorphic to $\SL_2(F)$; the stabilizer of $x_1$ has a parameter
bijection with $F^5\times\SL_2(F)$. The full stabilizer $P$ of the
line $Fx_1$ is the explicit parabolic, with a bijection
\[
 F^{\times}\times F^5\times\SL_2(F)\ \longrightarrow\ P.
\]
This bijection is a parameterization of the stabilizer. Since $|\SL_2(F)|=q(q^2-1)$, orbit--stabilizer gives
\[
 |G_2(F)|=\frac{q^6-1}{q-1}\,(q-1)q^5\,q(q^2-1),
\]
which is~\eqref{eq:g2order}. The order proof is independent of
perfectness, simplicity, and the binary exception.

\subsection{Singular-point geometry and simplicity}

The action on $\mathcal P$ is faithful. Its point stabilizer has four
orbits, with lengths
\begin{equation}\label{eq:g2subdegrees}
 1,\qquad q(q+1),\qquad q^3(q+1),\qquad q^5.
\end{equation}
Coordinate reductions identify these strata as orbits. Their geometry
also proves primitivity in every characteristic,
including $q=2$. The block-size argument leaves only the projective
points of $W=\langle x_1,x_2,x_3\rangle$ as a possible nontrivial
proper block. But the linear span of products $W\cdot W$ is the
distinguished line $Fx_1$. Any automorphism stabilizing $W$ must
therefore fix that point, contradicting transitivity of a block
stabilizer on a nontrivial block.

At $Fx_1$, a pair of commuting root groups gives an abelian subgroup
$R_0\cong(F^2,+)$ of order $q^2$, normal in the full point stabilizer.
Transport defines such a subgroup at every singular point, with
independence of the chosen transporter and covariance under $G_2(F)$.
For $q>2$, explicit root and Weyl relations show that their conjugates
generate the full group. The argument chooses a field element distinct
from $0$ and $1$ and applies also when $q=3$. Torus commutators then
give perfectness. The faithful primitive action and these local
abelian groups satisfy Proposition~\ref{prop:iwasawa}, proving simplicity.
Noncommutativity is established separately by two explicit root maps,
even at $q=2$. The positive simplicity proof does not use the numerical
order of $G_2(F)$.

\subsection{The binary exception}

For $F=\F_2$, let $H$ be the normal closure of the long-root
groups. The construction produces a second geometry: the
two-dimensional singular subspaces on which all products vanish.
There are exactly $63$ such zero-product planes. The count comes from
an exhaustive check of $64\times64$ binary coordinate pairs; the
group acting is still the full octonion automorphism group.

Two root involutions fix respectively $7$ and $9$ of these planes,
so their permutation signs are $(-1)^{28}=1$ and $(-1)^{27}=-1$.
The resulting sign homomorphism is surjective and contains $H$ in
its kernel. Independently, primitivity of the singular-point action
makes the nontrivial normal subgroup $H$ transitive, and the explicit
point-stabilizer generators give $[G_2(2):H]\leq2$. It follows that
$H$ is exactly the sign kernel. Using the already proved order,
\[
 |G_2(2)|=12096,\qquad [G_2(2):H]=2,\qquad |H|=6048.
\]
Field transport gives the same conclusion for every field of
cardinality two. This proves the negative case in the exact simplicity
theorem. The present package does not prove $H=G_2(2)'$, simplicity
of $H$, or its identification with $\PSU_3(3)$; the conventional
derived-group coincidence in Table~\ref{tab:exceptional-isos} remains
a separate comparison target.

\section{Small Ree groups from twisted matrices}\label{sec:ree}

Let $F$ be a finite field of characteristic three with
$q=|F|=3^{2m+1}$ and $m\geq1$. The model is the subgroup of $\GL_7(F)$ generated
by three explicit root maps, a diagonal torus, and a Weyl involution.
The coordinates follow B\"a\"arnhielm's explicit matrix model
\cite[Section~3]{BaarnhielmRee}; the geometric and simplicity arguments
draw on Wilson's elementary constructions~\cite{WilsonRee,WilsonReeAlgebra}.
Its order and simplicity are proved
for this generated group, without an abstract Ree-group theorem or
recognition by order.

\subsection{Field twists and the root group}

The twisting maps are the iterated Frobenius automorphisms
\[
 \theta(x)=x^{3^m},\qquad \sigma(x)=x^{3^{m+1}}.
\]
The finite-field identity $x^{3^{2m+1}}=x$ gives
$\theta\sigma=\sigma\theta=\mathrm{id}$ and
$\sigma^2(x)=x^3$. These identities, their compatibility with field
equivalences, and the bound $q\geq27$ are proved before the group
arguments.

Write $S(a,b,c)$ for the product of the three root matrices in their
specified order. Its entries recover $a,b,c$, and direct matrix
multiplication gives
\begin{align*}
 S(a,b,c)S(d,e,f)
  =S\bigl(&a+d,\ b+e-a\theta(d)^3,\\[-2pt]
          &c+f-db+a\theta(d)^3d-a^2\theta(d)^3\bigr).
\end{align*}
Together with the inverse formula, this identifies the root
subgroup $U$ with the parameter set $F^3$, so $|U|=q^3$.
Projection to the first parameter has abelian image and abelian kernel,
and hence $U$ is soluble. The diagonal subgroup $H\cong F^\times$
normalizes $U$, with conjugation on parameters
\[
 (a,b,c)\longmapsto
 \left(\frac{a\theta(\lambda)^3}{\lambda^2},
       \frac{b\lambda}{\theta(\lambda)^3},\frac{c}{\lambda}\right).
\]
Unique factorization in $B=UH$, the equality $U\cap H=1$, and the
split projection $B\to H$ with kernel $U$ prove that $B$ is soluble
and has order $q^3(q-1)$.

\subsection{The point geometry and the full stabilizer}

The Ree--Tits point set consists of a point at infinity and the
projective classes of the first rows of the matrices $S(a,b,c)$.
Projective parameter injectivity gives $q^3+1$ points, and $U$ acts
regularly on the affine points. The matrix convention uses row vectors;
inverse right multiplication supplies the corresponding left action.

The main geometric step is to prove preservation by the Weyl element
and to identify the full stabilizer of infinity. Seven exterior-kernel
relations and seven Tits-twisted bilinear equations characterize the
point chart. Their invariance under negative coordinate reversal proves
Weyl preservation. They also imply that the last affine coordinate
vanishes only at the affine origin. The generated group is consequently
doubly transitive.

Every generator preserves an explicit alternating cross product and the
twisted bilinear relation. At infinity the relation recovers a distinguished
two-dimensional subspace, and the cross product recovers the full flag.
A matrix fixing infinity must therefore be triangular. Applying the
opposite flag at the affine origin shows that the two-point stabilizer
is diagonal; the invariant tensors force it to be exactly $H$.
Regularity of $U$ then gives the full one-point stabilizer $B$.
This proves faithfulness as well as the order formula
\eqref{eq:reeorder} by orbit--stabilizer. The stabilizer calculation completes the order proof.

\subsection{Perfectness, simplicity, and comparison with octonions}

For $q\geq27$, torus commutators place all three root maps in the
derived subgroup. Short characteristic-three matrix identities then
place the Weyl involution and the minus-one torus element there.
Weyl inversion supplies all torus squares. Since $q\equiv3\pmod4$,
every nonzero scalar is a square or the negative of a square, so the
whole torus lies in the derived subgroup. The group is therefore perfect.

Simplicity uses a separately proved general criterion: a nontrivial
perfect group with a faithful quasiprimitive action and a soluble point
stabilizer is simple. Indeed, any nontrivial normal subgroup $N$ is
transitive, so the stabilizer surjects onto $G/N$. This quotient is both
soluble and perfect and is therefore trivial. The Ree action is
primitive, and its full stabilizer is the soluble subgroup $B$.
Explicit noncommuting root elements supply noncommutativity. This
simplicity argument does not use the numerical group-order theorem.

Entrywise field equivalence induces a group isomorphism of
the models. A separate comparison embeds the group into the
split-octonion $G_2(F)$ of Section~\ref{sec:g2}. In the fixed
seven-coordinate basis, the scalar--vector identification with the
eight-coordinate octonions is
\[
 (t,v)\longmapsto
 (v_0,v_1,v_2,t-v_3,t+v_3,-v_4,-v_5,-v_6).
\]
In characteristic three the octonion product becomes
\[
 (s,v)(t,w)=(st-B(v,w),\ sw+tv+v\mathbin{\times}w),
\]
where $B$ is the preserved symmetric form and $\times$ the preserved
cross product. Extending inverse row multiplication while fixing the
scalar coordinate thus gives an injective homomorphism into the full
octonion automorphism group. This uses structural octonion results,
not the order or simplicity of $G_2(F)$. No equality with the fixed
points of an exceptional graph automorphism is asserted.

Specializations of the uniform order theorem at $q=27,243,2187$
give, respectively,
\[
10\,073\,444\,472,\quad
49\,825\,657\,439\,340\,552,\quad
239\,189\,910\,264\,352\,349\,332\,632.
\]
These values specialize the uniform theorem to constructed finite fields. The public family excludes
$q=3$; the derived-group comparison with $\PSL_2(8)$ in
Table~\ref{tab:exceptional-isos} remains a separate task.

\section{Sporadic groups from retained structures}\label{sec:sporadic}

The sporadic branch follows connected constructions. The code used to
construct the Mathieu groups is retained in the Leech lattice, and the
resulting lattice and its full symmetry group are retained in later
stabilizer and scalar-centralizer constructions. Standard background
for these relationships is provided by~\cite{CS}; the formal models and
recorded completion results are those of the \atlas\ development
\cite{AtlasSoftware}. We first specify the shared objects, then describe
the individual groups in the order of Table~\ref{tab:sporadic}.
Their order and simplicity assertions are those of Formalization
result~\ref{thm:sporadics}.

\subsection{Shared code data: trio code, hexacode, and Golay code}\label{sec:golay-data}

The initial construction passes from a quadratic alphabet and its trio
code to the additive Kleinian hexacode, and then to the binary Golay code
by two explicit twists. The coordinate organization
\[
 3\ \longrightarrow\ 3\times2\ \longrightarrow\ (3\times2)\times4
\]
is retained along with markings, partitions, and recovery maps. The
quadratic functions on the intermediate alphabets are part of the data.
Any comparison with an $\F_4$-linear hexacode is a separate mathematical
statement.

The constructed binary code $C$ is a doubly even self-dual
$[24,12,8]$ code. Its weight-eight supports give the Witt design
$S(5,8,24)$. Write $\mathcal A=\Aut(C)$ for its full coordinate
automorphism group. Fix a dodecad $D$, a point $d\in D$, and two
distinct coordinates $i,j$; standard choices are supplied by the
construction. A subscript $D$ denotes a setwise dodecad stabilizer,
whereas coordinate subscripts denote pointwise stabilizers.

\subsection{The Mathieu groups}

\paragraph{\texorpdfstring{$M_{11}$}{M11}.}
The model is $(\mathcal A_D)_d$, the stabilizer of a point in the
dodecad stabilizer. It has order $7\,920$. Thus the entry is specified
directly by the marked Golay data; equivalently, it is the point
stabilizer in the dodecad model of $M_{12}$.

\paragraph{\texorpdfstring{$M_{12}$}{M12}.}
The model is $\mathcal A_D$, the setwise stabilizer of the chosen
dodecad, with its retained action on that dodecad. Its order is
$95\,040$.

\paragraph{\texorpdfstring{$M_{22}$}{M22}.}
The model is $\mathcal A_{i,j}$, fixing the two marked coordinates
individually, and has order $443\,520$. It is kept distinct from
the larger setwise stabilizer of the coordinate pair.

\paragraph{\texorpdfstring{$M_{23}$}{M23}.}
The model is $\mathcal A_i$, the stabilizer of one marked Golay
coordinate, and has order $10\,200\,960$.

\paragraph{\texorpdfstring{$M_{24}$}{M24}.}
The model is the full coordinate automorphism group $\mathcal A$.
Equality with the automorphism group of the octad design is proved
in the same ambient symmetric group. Sextet geometry supplies local
structure: every tetrad has a unique sextet completion, there are
$1771$ sextets, and the full sextet stabilizer has order $138240$.
The subsequent global action arguments give order $244823040$,
five-transitivity, and simplicity.

The actions on complements or dodecads, the corresponding residual
designs, and the embeddings relating the marked choices are retained
for the smaller Mathieu groups. Their catalogue order above does not
replace the common Golay construction by an induction from smaller
groups.

\subsection{Shared lattice data: the Leech lattice and its full isometry group}\label{sec:leech-data}

The retained Golay code defines an integral lattice $\Lambda$.
Its rank $24$, evenness, positive definiteness, unimodularity, minimum
squared norm $4$, and generation by minimal vectors are proved from the
construction. The group $\Co_0$ is its full linear isometry group.

The shell counts for squared norms $4,6,8$ are respectively
\[
 196560,\qquad 16773120,\qquad 398034000.
\]
They arise from exhaustive coordinate descriptions of lattice vectors.
The global order proof uses minimal vectors and orthogonal
pairs. An explicitly constructed nonmonomial isometry, together with
the full monomial subgroup, generates $\Co_0$. The resulting order is
\[
 |\Co_0|=8\,315\,553\,613\,086\,720\,000.
\]
The action on the $8292375$ intrinsic crosses, and its primitivity,
follow from the minimal-vector and orthogonal-pair analysis.

\subsection{The Conway groups}

\paragraph{\texorpdfstring{$\Co_1$}{Co1}.}
The simple model is the scalar quotient
\[
 \Co_1=\Co_0/\langle -I\rangle.
\]
It has order $4\,157\,776\,806\,543\,360\,000$. Finiteness and the
relevant kernel are proved before the final simplicity argument,
so the full lattice isometry group and its simple quotient remain
distinct objects.

\paragraph{\texorpdfstring{$\Co_2$}{Co2}.}
The model is the stabilizer in $\Co_0$ of a Leech vector of squared
norm $4$. Its order is $42\,305\,421\,312\,000$.

\paragraph{\texorpdfstring{$\Co_3$}{Co3}.}
The model is the stabilizer in $\Co_0$ of a Leech vector of squared
norm $6$. Its order is $495\,766\,656\,000$; additional retained
geometry includes its $276$-point action.

\subsection{The McLaughlin and Higman--Sims groups}

\paragraph{\texorpdfstring{$\mathrm{McL}$}{McL}.}
The model is the pointwise stabilizer of a Leech $2$--$2$--$3$
triangle and has order $898\,128\,000$. As in
Table~\ref{tab:sporadic}, the triangle label records half the
three squared edge lengths.

\paragraph{\texorpdfstring{$\mathrm{HS}$}{HS}.}
The model is the pointwise stabilizer of a Leech $2$--$3$--$3$
triangle and has order $44\,352\,000$. Its retained geometry
includes the rank-three $100$-point graph action.

The vector and triangle stabilizers retain their embeddings in the
lattice isometry group. Their lattice representations and stabilizer
descriptions allow these groups to be reused without reconstructing
unrelated permutation models.

\subsection{The Suzuki sporadic group \texorpdfstring{$\Suz$}{Suz}}

The constructions of $\Suz$ and $J_2$ enrich the same Leech lattice with
additional scalar structure. They illustrate why integral comparison
maps matter: a rational isometry between two ambient spaces must be
shown to map the lattices onto one another. The documented
comparisons prove both integral containments.

For $\Suz$, a constructed fixed-point-free isometry $\rho$ satisfies
\[
 \rho^3=I,\qquad \rho^2+\rho+I=0.
\]
Its full centralizer $C_E=C_{\Co_0}(\rho)$ is identified with the full
Hermitian isometry group of the corresponding Eisenstein lattice.
Its center consists of six scalar isometries. The constructed simple
quotient has order
\[
 |C_E|/6=2\,690\,072\,985\,600/6=448\,345\,497\,600.
\]
The action on $232960$ intrinsic frames and the thirteen suborbits of
a frame stabilizer are proved from the lattice. The sign quotient
$C_E/\{\pm I\}$ embeds into $\Co_1$ and maps onto $\Suz$ with central
kernel of order three. This keeps the embedded triple extension
distinct from its simple quotient.

\subsection{The Hall--Janko group \texorpdfstring{$J_2$}{J2}}

For $J_2$, let $\mathcal I$ be the constructed icosian order over
$\Z[(1+\sqrt5)/2]$, and let $C_H$ be the full quaternion-linear
Hermitian isometry group of the associated rank-three congruence
lattice. The underlying integral lattice is identified with the
retained Leech lattice. The $120$ norm-one icosian units act by right
scalars, and $C_H$ is their full centralizer in $\Co_0$.
The scalar action and the center of $C_H$ are separate objects.

There are $37800$ vectors of quaternionic Hermitian norm two, and
each associated right quaternionic line contains exactly $120$ of them.
Thus there are $315$ root lines. Each is orthogonal to ten others,
and every orthogonal pair has a unique completing root line. This
gives an intrinsic count of frames,
\[
 \frac{315\cdot10}{6}=525.
\]
The full coordinate-frame stabilizer has order $2304$. Explicit
quaternionic reflections generate the full group and prove frame
transitivity, whence
\begin{equation}\label{eq:j2order}
 |C_H|=525\cdot2304=1209600,\qquad
 |C_H/\{\pm I\}|=604800.
\end{equation}
The point stabilizer has suborbits of sizes
$1,10,32,32,80,160$. Their proved geometry yields primitivity.
Reflection relations yield perfectness, and the cyclic subgroup
generated by a point reflection supplies the abelian normal subgroup
in Proposition~\ref{prop:iwasawa}. The sign quotient is therefore
nonabelian simple. This example exhibits the chain from arithmetic and
local geometry to a full symmetry group, its order, and its simplicity.

\section{The Fischer constructions}\label{sec:fischer}

The development contains concrete models and proofs of order, simplicity,
finiteness, and noncommutativity for $\Fi_{22}$, $\Fi_{23}$,
and $\Fi_{24}'$. Their construction packages also retain the algebra,
reflecting rays, and residue geometry described below. All three
entries use one ambient ray group constructed from the algebra;
the order of presentation is not an inductive construction from the
smaller Fischer groups.

\subsection{The common algebra and ray group}

The starting point is the explicit Conway--Parker coordinate
algebra constructed from Golay and Parker data~\cite{FischerAlgebra}. The algebra has dimension $783$ over
$E=\mathbb Q(\omega)$, where $\omega$ is a primitive cube root of unity;
its coordinate index is the disjoint union of the 24 Golay coordinates
and the 759 octads. A proved tensor identity supplies the reflecting
automorphisms. Its normalized $\mu_3$-rays are three-element scalar orbits.
Let $F$ be the subgroup of permutations of the displayed reflecting
rays generated by the constructed root involutions.

The finite-algebra construction produces three families of rays of sizes
\[
 24,\qquad 759\cdot2^5,\qquad \binom{24}{2}2^{10},
\]
whose sum is $306936$. Frame geometry determines the full ray group,
which has a primitive rank-three action with subdegrees
$1,31671,275264$.

For marked commuting basic involutions $t_i$, a residue is obtained
from the centralizer of the marked set by dividing by the subgroup
it generates. The chosen marked involutions are part of the
construction. The order and simplicity assertions below are proved
for the resulting defined models; neither assertion is supplied
as a hypothesis.

\subsection{\texorpdfstring{$\Fi_{22}$}{Fi22}}

For two marked commuting basic involutions, the model is
\[
 G_{22}=C_F(t_i,t_j)/\langle t_i,t_j\rangle\quad(i\ne j).
\]
The quotient kernel has order four. The group has a faithful primitive
rank-three action on $3510$ points and proved order
\[
 |\Fi_{22}|=2^{17}3^9 5^2\,7\,11\,13.
\]
The construction retains the residue projection and its surjectivity,
generation by distinguished involutions, and perfectness, together
with the simplicity proof for this quotient.

\subsection{\texorpdfstring{$\Fi_{23}$}{Fi23}}

For one marked basic involution, the model is
\[
 G_{23}=C_F(t_i)/\langle t_i\rangle.
\]
The quotient kernel has order two. The group has a faithful primitive
rank-three action on $31671$ points and proved order
\[
 |\Fi_{23}|=2^{18}3^{13}5^2\,7\,11\,13\,17\,23.
\]
Here too the residue projection, its surjectivity, generation by
distinguished involutions, and perfectness are retained, together
with simplicity. The subdegrees of the two residue actions are
\[
\begin{array}{c|r|rrr}
 & \text{degree} & \multicolumn{3}{c}{\text{subdegrees}}\\
 G_{22} & 3510 & 1 & 693 & 2816\\
 G_{23} & 31671 & 1 & 3510 & 28160
\end{array}
\]
The completion record also retains the nonsplit central double cover
of $G_{22}$ inside $G_{23}$.

\subsection{\texorpdfstring{$\Fi_{24}'$}{Fi24 prime}}

The full three-transposition ray group $F$ is denoted $\Fi_{24}$.
The simple model $\Fi_{24}'$ is the kernel of its parity
homomorphism. This kernel is proved equal to $F'$, with $[F:F']=2$,
and has order
\[
 |\Fi_{24}'|=2^{21}3^{16}5^2 7^3\,11\,13\,17\,23\,29.
\]
The factor of two between $|\Fi_{24}|$ and $|\Fi_{24}'|$ is therefore
part of the derived-subgroup analysis. The proved action of the
simple group on the same $306936$ rays is faithful and transitive;
the recorded result does not include a rank-three assertion for
this subgroup.

The retained perfect central triple cover of $\Fi_{24}'$ acts
faithfully on the $783$-dimensional $E$-algebra. The quotient's
action is projective; no ordinary representation of the same degree
is inferred for it. In the full semilinear group, odd symmetries
invert the cubic scalars, so that scalar kernel is not central.
Thus the full ray group, its simple derived subgroup, and the
central triple cover remain distinct objects.

\section{The Monster and vertex operator algebras}\label{sec:monster}

H\"ohn--Seysen~\cite{MonsterOrder} determines the Monster and Baby
Monster orders directly from computations with axes in the
$196884$-dimensional Griess algebra, using Seysen's
\code{mmgroup} implementation~\cite{SeysenImplementation}. The Monster
is generated by a subgroup of shape $2^{1+24}_{+}.\mathrm{Co}_1$ and a
triality element. Orbit and stabilizer calculations use the Golay code,
Leech lattice, and Conway groups.

VOA arguments complement these calculations. The axes correspond, after
normalization, to Ising vectors in the Moonshine module $V^{\natural}$.
Orbifold and fusion arguments control involutions and their
centralizers, while the framed VOA structure identifies
$\operatorname{Aut}(V^{\natural})$ with the automorphisms of the Griess
algebra preserving its invariant bilinear form. Combining the computed
order with results of Carnahan and Borcherds, the paper proves that the
constructed Monster is this full automorphism group and has exactly two
conjugacy classes of involutions~\cite[Sections~2 and~5]{MonsterOrder}.
These are public mathematical results supporting future formalization;
the Monster and Baby Monster are among the outstanding \atlas\ entries.

This work continues the Moonshine programme of
Borcherds~\cite{BorcherdsVA} and Frenkel--Lepowsky--Meurman~\cite{FLM}.
Borcherds's proof of the Monstrous Moonshine conjectures
\cite{BorcherdsMoonshine} is a longer-term formalization goal, not a
result of the present development. The Baby Monster vertex operator
superalgebra provides a related source of sporadic symmetry
\cite{BabyVOA}. Further examples include the $10$-dimensional quadratic
space over $\F_2$ recovered from the modules of $V_{\sqrt2 E_8}^{+}$ in
Shimakura's treatment, with automorphism group $O_{10}^{+}(2)$
\cite[Theorem~4.5]{Shimakura}, and Griess--Lam's VOA existence proof of
the Monster~\cite{GriessLam}.

A VOA construction of the Fischer groups is being developed in
joint work with Lam~\cite{FischerVOA}. This manuscript in preparation
is separate from the Conway--Parker coordinate-algebra source
\cite{FischerAlgebra} used in Section~\ref{sec:fischer}.
The latter underlies the completed Fischer formalization; the VOA
construction is not included in the present release.

\section{Further development}\label{sec:programme}

Completion of the full sporadic catalogue is a longer-term objective.
The remaining eleven sporadic groups in Section~\ref{sec:inventory}
call for geometric constructions or verified computational arguments;
structural and comparison theorems for completed entries are a separate
continuing task.
The Hall--Janko, Fischer, Albert-algebra, and Cayley-plane source
manuscripts develop the geometric direction. For O'Nan's group,
Lee's 2025 account records the absence at that time of a computer-free
construction~\cite{LeeComputations}; Jansen--Wilson give explicit
modular constructions~\cite{ONan}. Formalizing such arguments requires
proofs of the algorithms and checks of the finite data they use.

A common Lie-type framework should construct groups from root data,
integral forms, and root operators, then establish generation, Bruhat
decomposition, centers, orders, and simplicity. Steinberg~\cite{Steinberg}
and Geck~\cite{Geck} provide sources. Beyond the completed families,
the targets are the unitary and minus-orthogonal families,
$F_4,E_6,E_7,E_8$, and the graph--field twists
$ {}^2E_6$ and $ {}^3D_4$.
The Suzuki and large Ree families
\[
 {}^2B_2(2^{2e+1}),\qquad {}^2F_4(2^{2e+1})\qquad(e\geq1)
\]
require exceptional characteristic-two endomorphisms. Their initial
nonsimple parameters and the Tits group $ {}^2F_4(2)'$ need separate
statements~\cite{Wilson}. Comparison with the existing $G_2$ and small
Ree models, including the exceptional fixed-point description of the
latter, is part of this programme.

For integration with classification arguments, priorities include
centralizers of prime-order elements, Sylow and local subgroups, outer
automorphisms, central extensions, and recognition geometries. The
comparison maps of Section~\ref{sec:classical} allow these results to
pass between the models encountered in different proofs.

Formalizing the Monster and Baby Monster results discussed in
Section~\ref{sec:monster} requires both verified finite computations and
infrastructure for graded spaces, vertex operations, lattice extensions,
and the relevant subalgebras or commutants. These developments would
connect the finite-group constructions to their vertex-algebraic
sources.

\section{Related formalization projects}\label{sec:related}

The Odd Order project of Gonthier and collaborators formalized the
Feit--Thompson theorem in Coq (now Rocq), using Mathematical
Components~\cite{OddOrder}. Completed in 2012 and described in 2013,
it proves that finite groups of odd order are soluble and supplies
extensive reusable finite-group and representation theory.
Mathlib supplies the cyclic, alternating, and rank-two PSL results
used here, together with group-action and simplicity
infrastructure~\cite{Mathlib,MathlibGroups}.

The Qiuzhen-CFSG project and FormaTheoria workflow report a Lean
development whose public README lists the Odd Order and Bender--Suzuki
theorems as completed~\cite{FormaTheoria}. Their reconstruction of
interdependent theories from the literature complements the construction
work here. This account follows the authors' public description and does
not report an independent rebuild of their project.

Independently, Yawara Ishida's \code{odd-order} project reports a
Lean 4/Mathlib formalization of the Feit--Thompson theorem, developed
from the textbooks of Isaacs, Bender--Glauberman, and P\'eterfalvi,
together with an extensive supporting finite-group and character-theory
library~\cite{IshidaOddOrder}. Its focus is the Odd Order Theorem and
its prerequisites rather than a catalogue of explicit simple-group
constructions. This description follows the public repository and does
not report an independent rebuild.

Two Utrecht theses from 2024 precede this work: Roxy van de Kuilen's
partial formalization of the first Janko group~\cite{Kuilen} and Erik
van der Plas's study of the Conway groups in Lean~\cite{Plas}.
Their precise declaration interfaces merit comparison with the
constructions presented here.

The public repository \nolinkurl{KitaKen1/finite-simple-groups-lean}
reports kernel-checked order and simplicity certificates for fourteen
sporadic permutation groups~\cite{KitaKen}:
\[
 M_{11},M_{12},M_{22},M_{23},M_{24},\mathrm{HS},J_2,J_1,
 \mathrm{McL},\Co_3,\Suz,\Co_2,\Fi_{22},\mathrm{He}.
\]
Its README also lists nine individual Lie-type groups, including
$G_2(4)$ and $\PSL_3(4)$. External programs generate candidate
certificates whose required properties are checked in Lean. This
description follows the public repository consulted on 18 September
2026, without an independent audit of its proofs.

Parameterized families require constructions and proofs for all
admissible ranks and fields. Root-group or geometric arguments provide
one route; a uniformly verified algorithm could provide another.
For individual groups, comparison isomorphisms could connect the
permutation models with the code, lattice, and algebra models here,
allowing their structural results to be combined.

\FloatBarrier

\appendix
\section{Lean implementation and mathematical interfaces}\label{sec:libraries}

\subsection{Lean and Mathlib inputs}

Lean supplies dependent type theory, elaboration, and the proof kernel;
Mathlib supplies mathematical definitions, theorems, and proof-producing
automation. In particular, \code{Group}, \code{Nat.card}, and
\code{IsSimpleGroup} are defined in
\leanname{Mathlib.Algebra.Group.Defs},
\leanname{Mathlib.SetTheory.Cardinal.Finite}, and
\leanname{Mathlib.GroupTheory.Subgroup.Simple}.
Appendix~\ref{sec:foundations} records their interpretation and the
logical axioms used.

Table~\ref{tab:Mathlib} lists representative mathematical inputs.
Declaration dependencies identify the portion used by a particular
proof; file imports make a broader interface available.

\begin{table}[!htbp]
\centering\small
\begin{tabularx}{\textwidth}{@{}>{\raggedright\arraybackslash}p{0.60\textwidth} X@{}}
\toprule
Representative modules (prefix \code{Mathlib.} omitted) & Mathematical input\\
\midrule
\leanname{Data.Fintype.Card}\newline
\leanname{SetTheory.Cardinal.Finite}\newline
\leanname{GroupTheory.Index}
& Finite types, cardinalities, subgroup indices, and counting identities.\\
\addlinespace
\leanname{GroupTheory.QuotientGroup.Basic}\newline
\leanname{GroupTheory.Subgroup.Centralizer}\newline
\leanname{GroupTheory.Commutator.Basic}
& Quotient groups, centralizers, commutators, and the maps used in central and derived-group constructions.\\
\addlinespace
\leanname{GroupTheory.GroupAction.Primitive}\newline
\leanname{GroupTheory.GroupAction.MultipleTransitivity}\newline
\leanname{GroupTheory.GroupAction.Iwasawa}
& Primitive and multiply transitive actions and the Iwasawa simplicity criterion.\\
\addlinespace
\leanname{GroupTheory.SpecificGroups.Cyclic}\newline
\leanname{GroupTheory.SpecificGroups.Alternating}\newline
\leanname{GroupTheory.SpecificGroups.Alternating.Simple}
& Existing cyclic and alternating group theory, including principal order and simplicity results.\\
\addlinespace
\leanname{Data.ZMod.Basic}\newline
\leanname{FieldTheory.Finite.GaloisField}
& Modular arithmetic and finite fields of prime-power order.\\
\addlinespace
\leanname{LinearAlgebra.Matrix.GeneralLinearGroup.Card}\newline
\leanname{LinearAlgebra.Matrix.SpecialLinearGroup}\newline
\leanname{LinearAlgebra.Matrix.Transvection}
& Matrix groups, general-linear-group cardinalities, and elementary matrices.\\
\addlinespace
\leanname{LinearAlgebra.Matrix.ProjectiveSpecialLinearGroup}\newline
\leanname{LinearAlgebra.Projectivization.Action}\newline
\leanname{LinearAlgebra.Projectivization.PSL.PSL2}
& Projective groups and actions, including the existing dimension-two simplicity theorem.\\
\addlinespace
\leanname{LinearAlgebra.Basis.Basic}\newline
\leanname{LinearAlgebra.BilinearForm.DualLattice}\newline
\leanname{LinearAlgebra.QuadraticForm.Basic}\newline
\leanname{Algebra.Quaternion}
& Bases, bilinear and quadratic forms, dual lattices, and quaternion arithmetic underlying the geometric constructions.\\
\bottomrule
\end{tabularx}
\caption{Representative mathematical inputs from the pinned Mathlib
dependency; they are
not an exhaustive or minimal import list.}
\label{tab:Mathlib}
\end{table}

The cyclic entry uses \leanname{ZMod.card} and
\leanname{isSimpleGroup_of_prime_card}; the alternating entry uses
\leanname{nat_card_alternatingGroup} and
\leanname{alternatingGroup.isSimpleGroup}.
For PSL, the existing \code{rank\_two\_simple} theorem and
\leanname{MulAction.IwasawaStructure.isSimpleGroup} support the
general-rank arguments in Section~\ref{sec:psl}.
\atlas\ develops the octonion multiplication and stabilizers,
Golay--Leech structures, and Conway--Parker algebra on this foundation.
Automation such as \code{ring} and \code{norm\_num} produces proof
terms checked by the kernel.

Shared mathematics can occur below public group interfaces. For example,
McL and HS use \leanname{Atlas.Sporadic.Conway3.card}, while the Fischer
entries share algebra and residue lemmas even when another Fischer
group's principal wrapper is absent from their dependencies.
The source-line footprints in Appendix~\ref{app:provenance} count
\atlas\ declarations; their derivations also use Lean and Mathlib.
The pinned environment is recorded there.

\subsection{Principal declaration interfaces}\label{app:interfaces}

The following identifiers connect the mathematical models described in
Sections~\ref{sec:elementary}--\ref{sec:fischer} to their formal
statements. They are collected here to keep the construction arguments
independent of implementation-level notation.

\paragraph{Classical comparisons.}
The comparisons in Tables~\ref{tab:uniform-isos}
and~\ref{tab:exceptional-isos} are Lean \code{MulEquiv} objects:
multiplication-preserving maps with proved inverses. The classical index
exports \code{b1EquivA1}, \code{b2EquivC2}, \code{d2EquivProduct},
\code{d3EquivA3}, and \code{evenBEquivC} under
\leanname{Atlas.Comparisons.Classical}.
The symplectic rank-one map is
\leanname{Atlas.Symplectic.rankOnePSL}; field equivalences induce
\code{bFieldEquiv} and \code{dFieldEquiv}.
The binary comparisons are
\leanname{Atlas.Orthogonal.b2BinaryEquivSymmetric} and
\leanname{Atlas.Orthogonal.b2BinaryDerivedEquivAlternating}.
The theorem
\leanname{Atlas.Comparisons.Classical.oddB_same_order_nonisomorphic_C}
records equal order and the nonexistence of a \code{MulEquiv}, using
the conjugacy-class invariance of $k_2$.

\paragraph{Split octonions.}
The model \leanname{Atlas.G2.Model} uses a generic construction of
multiplication-preserving linear automorphisms that requires no
associativity assumption.
The principal declarations in \leanname{Atlas.G2} are
\code{card\_Model}, \code{isSimple}, and \code{isSimple\_iff},
with \code{exists\_mul\_ne\_mul} for noncommutativity. The public
\leanname{Atlas.typeG2_construction} package collects the results for
the concrete model, and \leanname{Atlas.typeG2_prime_power} supplies
finite fields for every prime power. The corresponding
existence endpoints retain the concrete model as their witness.
The evidence for order, positive simplicity, the binary exception,
and the complete interface is separately scoped in
Appendix~\ref{sec:workflow}.
The exhaustive binary coordinate-pair check in
Section~\ref{sec:g2} uses ordinary kernel reduction.

\paragraph{Small Ree groups.}
The model is \leanname{Atlas.ReeG2.Model F m}.
The principal declarations are \leanname{Atlas.ReeG2.order} and
\leanname{Atlas.ReeG2.simple}. The public package
\leanname{Atlas.typeReeG2_construction} and existence theorem
\leanname{Atlas.exists_typeReeG2_parameter} use the same model.

\paragraph{Fischer groups.}
The declaration \leanname{Atlas.Fischer.rootGeneratedRayGroup}
defines the ray group $F$ of Section~\ref{sec:fischer}. The
\code{derived} and \code{index} theorems identify the parity kernel
with $F'$ and prove $[F:F']=2$.
For each Fischer entry, the public declarations \code{card}, \code{simple}, \code{finite}, and
\code{exists\_mul\_ne\_mul} establish the corresponding basic assertions.
The expanded signatures contain no supplied order or simplicity
hypothesis. The namespaces under \code{Atlas.Sporadic} are
\code{Fischer22}, \code{Fischer23}, and \code{Fischer24Prime}.
Construction and existence interfaces accompany these results.
The retained extension statements have their own statement and
dependency certificates.

\FloatBarrier

\section{Formal statements and logical foundations}\label{sec:foundations}

For a constructed carrier $G$, \code{Group G} supplies operations and
proofs of the group laws. Subgroups and quotients obtain these from
their ambient groups and proved normality. On finite types,
\code{Nat.card G} is the number of elements; it is zero on infinite
types. Finiteness is also exported as \code{Finite G}.
\code{IsSimpleGroup G} asserts nontriviality and that every normal
subgroup is trivial or the whole group. For the nonabelian entries,
noncommutativity is proved separately.

For example, the fixed model \leanname{Atlas.Sporadic.Fischer22.Model}
is the centralizer quotient $G_{22}$ of Section~\ref{sec:fischer}.
Its public statements are
\[
\begin{array}{ll}
 \code{finite}: & \code{Finite}\ G_{22},\\
 \code{card}: & \code{Nat.card}\ G_{22}=64561751654400,\\
 \code{simple}: & \code{IsSimpleGroup}\ G_{22},\\
 \code{exists\_mul\_ne\_mul}: & \exists x,y\in G_{22},\ xy\ne yx.
\end{array}
\]
All four declarations lie in \code{Atlas.Sporadic.Fischer22}; the
display suppresses printer-level universe and typeclass arguments.

The recorded dependency audits report only \code{propext}
(propositional extensionality), \code{Classical.choice}, and
\code{Quot.sound} (equality of related quotient representatives).
They exclude \code{sorryAx}, project-specific mathematical axioms, and
native-evaluator trust shortcuts. The proofs have been accepted by
Lean's standard kernel and, for the complete exported dependency
closure of the release selection, by the independently implemented
Nanoda checker (Appendix~\ref{sec:verification}). These checks still
rely on the foundations, the checking implementations, the export
pipeline, and the host system. The axiom report must be read together
with the full theorem type and the definitions of its carrier,
parameters, and instances: an axiom check alone does not establish
that the statement expresses the intended mathematics
\cite{LeanValidation,Soundness}.

\section{Development and recorded verification}\label{sec:workflow}

\subsection{Mathematical specification and implementation}

The author specifies the objects, theorems, proof architecture, and
permitted dependencies. Codex, OpenAI's coding agent, searches the
pinned library and implements and repairs proofs against Lean's
elaborator and kernel. The author used ChatGPT, including
GPT-6 Pro, for mathematical discussion and specifications, and Codex
with GPT-6 Astra for Lean implementation. These names identify the
development tools used; they are not a controlled comparison of
models or a measurement of autonomous performance.

The author directs and maintains \atlas\ and is responsible for its
mathematical claims and formal specifications. Generative AI also
assisted the drafting and revision of this article; the author is
responsible for its final wording and mathematical content. Mathematical sources
and joint work are credited where used. Compiled statements and
dependency reports support review of the results; reproducing the
AI's search is unnecessary for checking the resulting proof terms.

\subsection{Recorded verification}\label{sec:verification}

The completed runs of 18 September 2026 concern release commit
\code{935474b} and a selection of 911 declarations covering catalogue
roles, structural properties, and comparison interfaces for the eight
families and fifteen sporadic entries. The selection counts declarations,
not groups. The checks address the formal statements, their dependencies,
and acceptance by two checker implementations; the precise source state
and checking details are given in Appendix~\ref{app:provenance}.

\paragraph{Build and dependency audit.}
The project-source build completed successfully, with 5\,677 Lake jobs.
The compiled-statement and transitive dependency audit covered 911 roots
and 70\,266 nodes, including type and proof dependencies, and found only
the three permitted logical axioms of Appendix~\ref{sec:foundations}.
The build used pinned library caches, rather than recompiling every
Mathlib source file.

\paragraph{Independent proof replay.}
Nanoda, an independently implemented checker written in Rust, accepted
all 70\,891 declarations in the complete exported dependency closure of
the 911 selected roots, with no errors. This includes the library
dependencies required by those roots, not unrelated imported theorems
or all of Mathlib. Nanoda's exported-declaration count and the dependency
audit's node count are produced by different tools and are reported
separately.

\paragraph{Comparator.}
The same 911 targets passed Comparator: 869 theorem targets and 42
non-theorem targets represented by equality wrappers, with no definition
holes. Statement and definition comparisons, permitted-axiom checks,
standard Lean kernel replay, and the quotient post-check all succeeded.
The Challenge and Solution import the same frozen release. This checks
agreement with that formal reference and replay acceptance, rather than
comparison against an independently authored mathematical specification.
Comparator used the standard Lean kernel; Nanoda was a separate run.

\paragraph{Statement consistency.}
Additional inspections of the compiled order and simplicity statements
confirmed that each pair uses the same model, parameters, and inherited
group law. For the eight families, admissibility predicates contain the
intended numerical restrictions rather than assumed conclusions, and
$q$ is the cardinality of the same field $F$ wherever field parameters
occur. The Lie-rank and matrix-dimension conventions agree. For the
fifteen sporadic entries, the inspections also checked Mathieu markings
and the distinctions between stabilizers, simple quotients, covering
groups, and the $\Fi_{24}'$ parity kernel. These are inspections of the
carrier applications and hypotheses, not additional kernel runs or
external recognition theorems.

The earlier development catalogue check validated 886 declaration
references, including aliases and structural exports. The source-footprint
measurement used a separate selection of 46 principal order and
simplicity declarations for 23 entries, with a complete graph of
65\,676 dependency nodes and project source ranges recovered throughout.
Table~\ref{tab:footprint} retains that measurement. Neither the broader
release audit nor the replay counts replace this line-count calculation.
Family, geometric, and comparison audits use their own selections.

In particular, the $G_2$ and small Ree order audits exclude their
perfectness and simplicity arguments. Their positive simplicity
audits exclude the numerical group-order theorems. The Ree proofs
also exclude untwisted $G_2$ results; its octonion embedding uses a
separate structural scope.
Finite tables checked by ordinary kernel reduction are part of some proofs.

\subsection{Repository and checking instructions}\label{sec:release}

The source repository is
\mbox{\url{https://github.com/Moonshine-in-Kansas/atlas}}.
The release catalogue is \mbox{\leanname{docs/catalogue.json}}; the
build and dependency-audit instructions are \code{AUDIT.md} and
\mbox{\leanname{verification/check.rb}}. Appendix~\ref{app:provenance}
gives the tool revisions, checking commands, and scope and resource use
of the completed runs.

The release's default license is the \emph{ATLAS Research and Attribution
License 1.0}, with third-party rights preserved. The license and
third-party notices accompany the source.

\section{Source state and reproducibility}\label{app:provenance}

The recorded build, dependency audit, Nanoda replay, Comparator run,
and statement inspections concern the release snapshot at commit
\leanname{935474b0c3678e053e1152a49ed6593bd8ad44c9}
(18 September 2026)~\cite{AtlasSoftware}. All 2\,299 recorded
project-source and toolchain files matched the checked snapshot when
the Comparator success and statement inspections were recorded.

The line measurements below refer instead to the development snapshot
of 17 September 2026, 07:34:33 UTC. The recorded verification runs do not
remeasure this earlier inventory or its declaration footprints.

\subsection{Source inventory}

The measured development inventory contains 2\,295 tracked Lean files under
\code{Atlas/}, with 161\,772 lines, and the root driver
\code{Atlas.lean}, with 706 lines. The total is therefore
\textbf{162\,478 lines in 2\,296 files}. This includes audit and helper
modules, comments, and blank lines; it excludes untracked drafts,
Mathlib, and verification-generated Lean files. The small Ree-specific
modules, including finite-field support and the public package, account
for 54 files and 3\,470 lines, without shared generic infrastructure.

The physical file sizes are distinct from the declaration-range
footprints in Table~\ref{tab:footprint}, which count only source lines
used by the selected principal results.

\subsection{Build and dependency audit}

The environment is pinned to Lean \code{leanprover/lean4:v4.34.0-rc2},
compiler commit
\leanname{6a10ac8c22beadecabdbb0919c2b50214762f91d}, and Mathlib revision
\leanname{85e3a25e006c35636f0e53b0e9296caca2685bc0}.
From the release root, with that toolchain installed, run
\begin{verbatim}
lake exe cache get
LEAN_NUM_THREADS=2 lake build
LEAN_NUM_THREADS=2 ruby verification/check.rb
\end{verbatim}
The first command fetches the pinned dependency cache. The remaining
commands build the configured \atlas\ targets and run the release
audit; they do not invoke Nanoda or Comparator. Preserve
\code{lean-toolchain}, \code{lakefile.toml}, and \code{lake-manifest.json}.
The release catalogue \mbox{\leanname{docs/catalogue.json}} maps claims
to declarations and evidence, and \code{AUDIT.md} documents the build
and dependency checks.

The final build/audit driver ran from 05:43:04 to
05:52:59 UTC on 18 September. This is the final driver's elapsed
interval, not the time for a source rebuild from scratch. A preceding
attempt had a process exit with status 137. The affected target was
then built successfully on its own, and the final full driver passed.

\subsection{Replay and statement-inspection results}\label{app:replays}

\paragraph{Nanoda.}
The checker revision is
\leanname{4c544ed4099c8227f07d5de77ad1e69fb0740a27}.
The run checked 70\,891 exported declarations, exited with status zero,
and finished at 06:40:26 UTC on 18 September. Checking took
31 minutes 0.82 seconds on one thread, with maximum resident memory
5\,375\,116 KiB (about 5.13 GiB); export took 93.66 seconds.
The exported data occupied 936\,181\,371 bytes.

\paragraph{Comparator.}
The run used Comparator revision
\leanname{2312244ac716564a61cc0bf4e107d9abf1757a61}, exporter revision
\leanname{cacf989bd75f608700820f6afc595f32e7a99a4d}, and the Landrun sandbox
at revision \leanname{811cfff51ceaf3d9843708aa6d22e9b84ccac8b4}.
It exited with status zero after 41 minutes 35.542 seconds of service
runtime. The formal-reference qualification is stated in
Appendix~\ref{sec:verification}.

\paragraph{Statement inspections.}
The actual carrier applications were inspected in addition to the
weaker automatic test for the presence of model declarations in type
dependencies.

\subsection{Per-group declaration footprints}\label{app:footprints}

Table~\ref{tab:footprint} measures the source needed by the selected
order and simplicity declarations for each completed entry. The global
deduplicated \atlas\ footprint is \textbf{88\,856 lines}. This measures
source dependencies, not proof-term size, model-token usage, or the
amount of new mathematics.

Each column counts physical source lines in compiler-recorded
declaration ranges reached through the selected theorem's transitive
type and proof dependencies, including interior comments and blank
lines. Generated helpers use the recorded enclosing declaration's
range; a final position at column zero excludes that line. Ranges are
merged by file, so column and row totals overlap. Lean core and external
libraries are excluded.

For $\PSL_n(q)$, $n$ is matrix dimension; $r$ is rank for $B,C,D$.
The type-$A$, symplectic, type-$B$, and $G_2$ simplicity roots are
exact-range equivalences including their exceptional cases. The Ree
root applies for $q=3^{2m+1}$, $m\geq1$. Isomorphic construction
families have separate rows.

The 46-root footprint audit has 65\,676 nodes, complete graph references,
and resolved source ranges for every project dependency. The broader
911-declaration release audit is a separate verification scope.

\begin{table}[!htbp]
\centering\small
\begin{tabular}{@{}lrrr@{}}
\toprule
Group or family & Order footprint & Simplicity footprint & Union \\
\midrule
$C_p$ & 5 & 12 & 12 \\
$\Alt_n$ & 4 & 4 & 7 \\
$\PSL_n(q)$ & 168 & 322 & 379 \\
$B_r(q)$ & 4\,121 & 7\,898 & 8\,009 \\
$C_r(q)$ & 708 & 2\,203 & 2\,203 \\
$D_r(q)$ & 4\,290 & 2\,466 & 5\,408 \\
$G_2(q)$ & 1\,237 & 3\,287 & 3\,568 \\
${}^2G_2(q)$ & 1\,570 & 2\,225 & 2\,250 \\
$\mathrm{M}_{11}$ & 5\,479 & 5\,866 & 5\,868 \\
$\mathrm{M}_{12}$ & 4\,924 & 5\,959 & 5\,961 \\
$\mathrm{M}_{22}$ & 4\,183 & 5\,348 & 5\,351 \\
$\mathrm{M}_{23}$ & 4\,119 & 4\,492 & 4\,495 \\
$\mathrm{M}_{24}$ & 3\,477 & 4\,582 & 4\,585 \\
$\mathrm{Co}_{1}$ & 10\,826 & 12\,334 & 12\,346 \\
$\mathrm{Co}_{2}$ & 10\,825 & 14\,070 & 14\,072 \\
$\mathrm{Co}_{3}$ & 12\,137 & 16\,423 & 16\,423 \\
$\mathrm{McL}$ & 13\,704 & 16\,375 & 16\,377 \\
$\mathrm{HS}$ & 13\,447 & 14\,771 & 14\,776 \\
$\mathrm{Suz}$ & 19\,727 & 21\,320 & 21\,374 \\
$\mathrm{J}_{2}$ & 6\,989 & 7\,954 & 8\,184 \\
$\mathrm{Fi}_{22}$ & 31\,955 & 34\,248 & 34\,250 \\
$\mathrm{Fi}_{23}$ & 31\,950 & 34\,247 & 34\,249 \\
$\mathrm{Fi}_{24}'$ & 31\,813 & 33\,396 & 33\,397 \\
\midrule
\textbf{Global deduplicated union} & & & \textbf{88\,856} \\
\bottomrule
\end{tabular}
\caption{Compiler-derived \atlas\ source-line dependency footprints for
46 principal order and simplicity declarations: eight families and
fifteen sporadic constructions. Shared source lines are counted once
within each union; different rows need not have disjoint dependencies.}
\label{tab:footprint}
\end{table}

\FloatBarrier


\begingroup\small
\begin{thebibliography}{99}
\interlinepenalty=10000
\setlength{\itemsep}{0pt}
\setlength{\parskip}{0pt}

\bibitem{BaarnhielmRee}
H. B\"a\"arnhielm, Recognising the small Ree groups in their natural
representations, \emph{Journal of Algebra} \textbf{416} (2014), 139--166.
\href{https://doi.org/10.1016/j.jalgebra.2014.06.017}{doi:10.1016/j.jalgebra.2014.06.017};
\href{https://arxiv.org/abs/1206.0411}{arXiv:1206.0411}.

\bibitem{BorcherdsVA}
R.~E. Borcherds,
Vertex algebras, Kac--Moody algebras, and the Monster,
\emph{Proceedings of the National Academy of Sciences USA}
\textbf{83} (1986), 3068--3071.
\href{https://doi.org/10.1073/pnas.83.10.3068}{doi:10.1073/pnas.83.10.3068}.

\bibitem{BorcherdsMoonshine}
R.~E. Borcherds,
Monstrous moonshine and monstrous Lie superalgebras,
\emph{Inventiones Mathematicae} \textbf{109} (1992), 405--444.
\href{https://doi.org/10.1007/BF01232032}{doi:10.1007/BF01232032}.

\bibitem{BorovikInvolutions}
A.~V. Borovik, \emph{Orthogonal and symplectic black box groups,
revisited}, 2001, Sections~2.5--2.6.
\href{https://arxiv.org/abs/math/0110234}{arXiv:math/0110234}.

\bibitem{ATLAS85}
J.~H. Conway, R.~T. Curtis, S.~P. Norton, R.~A. Parker, and R.~A. Wilson,
\emph{Atlas of Finite Groups: Maximal Subgroups and Ordinary Characters
for Simple Groups}, Clarendon Press, Oxford, 1985.

\bibitem{CS}
J.~H. Conway and N.~J.~A. Sloane,
\emph{Sphere Packings, Lattices and Groups}, 3rd ed., Grundlehren der
mathematischen Wissenschaften 290, Springer, New York, 1999.
\href{https://doi.org/10.1007/978-1-4757-6568-7}{doi:10.1007/978-1-4757-6568-7}.

\bibitem{FLM}
I. Frenkel, J. Lepowsky, and A. Meurman,
\emph{Vertex Operator Algebras and the Monster},
Pure and Applied Mathematics 134, Academic Press, 1988.

\bibitem{Geck}
M. Geck, \emph{A Course on Lie algebras and Chevalley groups},
2024; revised version~3, 2 October 2025.
\href{https://arxiv.org/abs/2404.11472v3}{arXiv:2404.11472v3}.

\bibitem{OddOrder}
G. Gonthier et al., A machine-checked proof of the Odd Order Theorem,
in \emph{Interactive Theorem Proving}, LNCS 7998, Springer, 2013,
pp.~163--179.
\href{https://doi.org/10.1007/978-3-642-39634-2_14}{doi:10.1007/978-3-642-39634-2\_14}.
Source: \url{https://github.com/math-comp/odd-order}.

\bibitem{GriessLam}
R.~L. Griess, Jr. and C.~H. Lam,
A new existence proof of the Monster by VOA theory,
\emph{Michigan Mathematical Journal} \textbf{61} (2012), 555--573.
\href{https://doi.org/10.1307/mmj/1347040259}{doi:10.1307/mmj/1347040259};
\href{https://arxiv.org/abs/1103.1414}{arXiv:1103.1414}.

\bibitem{GMS}
R.~L. Griess, Jr., U. Meierfrankenfeld, and Y. Segev,
A uniqueness proof for the Monster,
\emph{Annals of Mathematics} \textbf{130} (1989), 567--602.

\bibitem{AtlasSoftware}
G. H\"ohn, \emph{ATLAS: constructive finite simple groups in Lean},
source release, commit \code{935474b},
18 September 2026.
\mbox{\url{https://github.com/Moonshine-in-Kansas/atlas}}.
Source identification, development measurements, and checking scope
are given in Appendix~\ref{app:provenance}.

\bibitem{FischerAlgebra}
G. H\"ohn,
\emph{The Conway--Parker algebra and the largest Fischer group},
preprint, September 2026; mathematical source for
the completed Fischer algebra formalization.

\bibitem{HallJankoNote}
G. H\"ohn,
\emph{The Hall--Janko Simple Group and the Icosian Leech Lattice},
preprint, 2026.
\href{https://arxiv.org/abs/2609.22842}{arXiv:2609.22842}.

\bibitem{AlbertMassNote}
G. H\"ohn,
\emph{Definite integral Albert algebras: arithmetic and finite geometry},
preprint, 2026.
\href{https://arxiv.org/abs/2609.27550}{arXiv:2609.27550}.

\bibitem{BabyVOA}
G. H\"ohn, \emph{Selbstduale Vertexoperatorsuperalgebren und das
Babymonster}, Bonner Mathematische Schriften 286, 1996.
\href{https://arxiv.org/abs/0706.0236}{arXiv:0706.0236}.

\bibitem{CayleyPlaneNote}
G. H\"ohn,
\emph{The unique extremal threemodular lattice of rank 26, the
generalized hexagon $(2,8)$, and the tight Cayley-plane $5$-design},
preprint, 2026.
\href{https://arxiv.org/abs/2609.25086}{arXiv:2609.25086}.

\bibitem{FischerVOA}
G. H\"ohn and C.-H. Lam,
\emph{A vertex operator algebra construction of the Fischer groups},
working title; manuscript in preparation.

\bibitem{MonsterOrder}
G. H\"ohn and M. Seysen,
\emph{The Order of the Monster Finite Simple Group}, 2025.
\href{https://arxiv.org/abs/2508.01037}{arXiv:2508.01037}.

\bibitem{IshidaOddOrder}
Y. Ishida,
\emph{odd-order: The Feit--Thompson Odd Order Theorem in Lean 4},
software repository and README, 2026; consulted 24 September 2026.
\url{https://github.com/yawara/odd-order}.

\bibitem{ONan}
C. Jansen and R.~A. Wilson, Two new constructions of the O'Nan group,
\emph{Journal of the London Mathematical Society} \textbf{56} (1997),
579--583.
\href{https://doi.org/10.1112/S0024610798005742}{doi:10.1112/S0024610798005742}.

\bibitem{OrderRecognition}
W. Kimmerle, R. Lyons, R. Sandling, and D.~N. Teague,
Composition factors from the group ring and Artin's theorem on orders
of simple groups, \emph{Proceedings of the London Mathematical Society}
(3) \textbf{60} (1990), 89--122.
\href{https://doi.org/10.1112/plms/s3-60.1.89}{doi:10.1112/plms/s3-60.1.89}.

\bibitem{KitaKen}
KitaKen1,
\emph{FiniteSimpleGroups: Lean/mathlib formalizations of finite simple
groups}, software repository and README, consulted 18 September 2026.
\url{https://github.com/KitaKen1/finite-simple-groups-lean}.

\bibitem{Kuilen}
R. van de Kuilen, \emph{The first Janko group $J_1$: simplicity and
formalization}, bachelor's thesis, Utrecht University, 2024.
\url{https://studenttheses.uu.nl/items/7f8263cb-453b-4e78-99e0-35135037c33e}.

\bibitem{LeeComputations}
M. Lee, Computation and the sporadic simple groups,
in \emph{Computational Group Theory}, Oberwolfach Report 27/2025,
pp.~1411--1413.
\url{https://publications.mfo.de/bitstream/handle/mfo/4340/OWR_2025_27.pdf}.

\bibitem{Mathlib}
The mathlib Community,
The Lean mathematical library,
in \emph{Proceedings of the 9th ACM SIGPLAN International Conference
on Certified Programs and Proofs (CPP 2020)}, ACM, 2020,
pp.~367--381.
\href{https://doi.org/10.1145/3372885.3373824}{doi:10.1145/3372885.3373824};
\href{https://arxiv.org/abs/1910.09336}{arXiv:1910.09336}.

\bibitem{MathlibGroups}
The mathlib Community, \emph{Mathlib group-theory documentation},
including alternating groups and projective special linear groups.
\url{https://leanprover-community.github.io/mathlib4_docs/Mathlib/GroupTheory/SpecificGroups/Alternating.html}.

\bibitem{LeanValidation}
The Lean development team,
\emph{Validating a Lean Proof}, Lean Language Reference,
online documentation, consulted 18 September 2026.
\url{https://lean-lang.org/doc/reference/latest/ValidatingProofs/}.

\bibitem{Soundness}
L. de Moura,
\href{https://leodemoura.github.io/blog/2026-8-24-postmortem-for-the-kernel-soundness-bug-hunt/}{\emph{Postmortem for the Kernel Soundness Bug Hunt}},
24 August 2026, online article.

\bibitem{Lean4}
L. de Moura and S. Ullrich,
The Lean 4 theorem prover and programming language,
in \emph{Automated Deduction---CADE 28}, Lecture Notes in Computer
Science 12699, Springer, 2021, pp.~625--635.
\href{https://doi.org/10.1007/978-3-030-79876-5_37}{doi:10.1007/978-3-030-79876-5\_37}.

\bibitem{FormaTheoria}
T. Nie, A. Zhang, Y. Tang, D. Testa, S.-T. Yau, P. Li, and Y. Zhou,
\emph{FormaTheoria: Constructing Large-Scale Lean Theories from
Mathematical Literature---Toward the Formalization of the Classification
of Finite Simple Groups}, 2026.
\href{https://arxiv.org/abs/2608.10894}{arXiv:2608.10894}.
Source and project-status description:
\url{https://github.com/Qiuzhen-CFSG/CFSG}, consulted 18 September 2026.

\bibitem{Plas}
E. van der Plas, \emph{Formalisation of the Finite Simple Conway Groups
in Lean}, bachelor's thesis, Utrecht University, 2024.
\url{https://studenttheses.uu.nl/handle/20.500.12932/46756}.

\bibitem{SeysenImplementation}
M. Seysen,
A fast implementation of the Monster group: The Monster has been tamed,
\emph{Journal of Computational Algebra} \textbf{9} (2024), Article 100012.
\href{https://doi.org/10.1016/j.jaca.2024.100012}{doi:10.1016/j.jaca.2024.100012}.

\bibitem{Shimakura}
H. Shimakura,
The automorphism group of the vertex operator algebra $V_L^+$ for an
even lattice $L$ without roots,
\emph{Journal of Algebra} \textbf{280} (2004), 29--57.
\href{https://doi.org/10.1016/j.jalgebra.2004.05.018}{doi:10.1016/j.jalgebra.2004.05.018};
\href{https://arxiv.org/abs/math/0311141}{arXiv:math/0311141}.

\bibitem{Steinberg}
R. Steinberg, \emph{Lectures on Chevalley Groups},
University Lecture Series 66, American Mathematical Society, 2016;
original Yale lecture notes, 1967--1968.

\bibitem{Taylor}
D.~E. Taylor, \emph{The Geometry of the Classical Groups},
Sigma Series in Pure Mathematics 9, Heldermann Verlag, Berlin, 1992.
Author's corrected text:
\url{https://www.maths.usyd.edu.au/u/don/papers/gcg.pdf}.

\bibitem{WilsonReeAlgebra}
R.~A. Wilson, Another new approach to the small Ree groups,
\emph{Archiv der Mathematik} \textbf{94} (2010), 501--510.
\href{https://doi.org/10.1007/s00013-010-0130-4}{doi:10.1007/s00013-010-0130-4}.

\bibitem{WilsonRee}
R.~A. Wilson, A new construction of the Ree groups of type ${}^2G_2$,
\emph{Proceedings of the Edinburgh Mathematical Society}
\textbf{53} (2010), 531--542.
\href{https://doi.org/10.1017/S001309150800028X}{doi:10.1017/S001309150800028X}.
Author's text:
\url{https://webspace.maths.qmul.ac.uk/r.a.wilson/pubs_files/ReeG2.pdf}.

\bibitem{Wilson}
R.~A. Wilson,
\emph{The Finite Simple Groups}, Graduate Texts in Mathematics 251,
Springer, London, 2009.
\href{https://doi.org/10.1007/978-1-84800-988-2}{doi:10.1007/978-1-84800-988-2}.

\bibitem{ATLASOnline}
R.~A. Wilson and collaborators,
\emph{ATLAS of Finite Group Representations}, online database;
classical and exceptional small-group identifications consulted
15 September 2026.
\url{https://brauer.maths.qmul.ac.uk/Atlas/}.

\end{thebibliography}
\endgroup
\end{document}